\documentclass[11pt,a4]{article}
\usepackage{fullpage}
\usepackage{times}

\usepackage{amsmath,amsfonts,graphpap,amscd,mathrsfs,graphicx,lscape}
\usepackage{epsfig,amssymb,amstext,xspace}
\usepackage{theorem}
\usepackage{algorithm,algorithmic}

\usepackage{epsfig}
\usepackage{bbm, url}

\newenvironment{proof}{\bg{Proof.}}{\ed}

\newcommand{\E}{\mathbb{E}}
\newcommand{\abs}[1]{\vert{#1}\vert}

\newtheorem{theorem}{Theorem}[section]\newtheorem{lemma}[theorem]{Lemma}

\newtheorem{claim}[theorem]{Claim}

\newtheorem{defn}[theorem]{Definition}

\newtheorem{fact}[theorem]{Fact}

{\theorembodyfont{\rmfamily} \newtheorem{remark}[theorem]{Remark}}

\def\squareforqed{\hbox{\rule{2.5mm}{2.5mm}}}

\def\QED{\ifmmode\squareforqed 
  \else{\nobreak\hfil   
    \penalty50                 
    \hskip1em                  
    \null                      
    \nobreak                   
    \hfil                      
    \squareforqed              
    \parfillskip=0pt           
    \finalhyphendemerits=0     
    \endgraf}                  
  \fi}

\def\blksquare{\rule{2mm}{2mm}}
\def\qedsymbol{\blksquare}
\newcommand{\bg}[1]{\medskip\noindent{\bf #1}}
\newcommand{\ed}{{\hfill\qedsymbol}\medskip}
\newenvironment{proofsk}{\noindent\textbf{Proof Sketch.}}{\QED}
\newenvironment{proofof}[1]{\bg{Proof of #1.}}{\ed}

\newcommand{\R}{\ensuremath{\mathbb R}}

\newcommand{\comment}[1]{}

\newcommand{\junk}[1]{}

\newcommand{\ud}{\,\mathrm{d}}

\newcommand{\prob}{\mathbb{P}}

\newlength{\tmp} \newlength{\lpsx} \newlength{\lpsy} \newlength{\upsx} \newlength{\upsy}

\newcommand{\one}[1]{\mathbbm{1}\{#1\}}

\newcommand{\bid}{b}
\newcommand{\bids}{{\mathbf \bid}}
\newcommand{\bidsmi}{{\mathbf \bid}_{-i}}
\newcommand{\bidi}[1][i]{{\bid_{#1}}}
\newcommand{\hs}{{\mathbf h}}
\newcommand{\val}{v}
\newcommand{\vals}{{\mathbf \val}}
\newcommand{\valsmi}{{\mathbf \val}_{-i}}
\newcommand{\vali}[1][i]{{\val_{#1}}}

\newcommand{\util}{u}

\newcommand{\utili}[1][i]{{\util_{#1}}}

\newcommand{\dist}{F}
\newcommand{\dists}{{\mathbf \dist}}

\newcommand{\price}{p}

\newcommand{\pricei}[1][i]{{\price_{#1}}}

\newcommand{\ts}{{\mathbf t}}
\newcommand{\tsmi}{\ts_{-i}}

\newcommand{\signal}{{\mathbf s}}

\begin{document}

\title{Bounding the inefficiency of outcomes
in generalized second price auctions\thanks{Preliminary versions of some of the results in this paper
appeared in
Paes Leme and Tardos, FOCS'10 \cite{PLT10}, Lucier and Paes Leme,
EC'11 \cite{LPL11}, and Caragiannis, Kaklamanis, Kanellopoulos and
Kyropoulou, EC'11 \cite{CKKK11}.}}

\author{Ioannis Caragiannis\thanks{Computer Technology Institute and Press ``Diophantus'' \& Department of Computer Engineering and Informatics, University of Patras, Greece. Email: \texttt{\{caragian,kakl,kanellop\}@ceid.upatras.gr}}
\and Christos Kaklamanis$^\dag$
\and Panagiotis Kanellopoulos$^\dag$
\and Maria Kyropoulou\thanks{Department of Computer Science, University of Oxford, UK. Email: \texttt{kyropoul@cs.ox.ac.uk}}
\and Brendan Lucier\thanks{Microsoft Research New England, USA. Email: \texttt{brlucier@microsoft.com}}
\and Renato Paes Leme\thanks{Google Research New York, USA.  \comment{Part of
this work was conducted while at Department of Computer Science.}  Email:
\texttt{renatoppl@google.com}}
\and \'{E}va Tardos\thanks{Corresponding author. Department of Computer Science, Cornell University, USA. Supported in part by NSF grants CCF-0910940 and CCF-0729006, ONR grant N00014-98-1-0589,  a Yahoo!~Research Alliance Grant, and a Google Research Grant. Email: \texttt{eva@cs.cornell.edu}}}
\date{}

\maketitle
\begin{abstract}
The Generalized Second Price (GSP) auction is the primary auction used for
monetizing the use of the Internet. It is well-known that truthtelling is
not a dominant strategy in this auction and that inefficient equilibria can
arise. Edelman et al. (AER, 2007) and Varian (IJIO, 2007) show that an efficient
equilibrium always exists in the full information setting. Their results,
however, do not extend to the case with uncertainty, where efficient equilibria
might not exist.

In this paper we study the space of equilibria in GSP, and quantify the efficiency loss
that can arise in equilibria under a wide range of sources of uncertainty, as well as in the full information setting. The traditional Bayesian game models uncertainty in the valuations (types) of the participants. The Generalized Second Price (GSP) auction gives rise to a further form of uncertainty: the selection of quality factors resulting in uncertainty about the behavior of the underlying ad allocation algorithm.
The bounds we obtain apply to both forms of uncertainty, and are robust in the
sense that they apply
under various perturbations of the solution concept, extending to models with information asymmetries and bounded rationality in the form
of learning strategies.

We present a constant bound ($2.927$) on the factor of the efficiency loss
(\emph{price of anarchy}) of
the corresponding game for the Bayesian model of partial information
about other participants and about ad quality factors.
For the full information setting, we prove a surprisingly
low upper bound of $1.282$
on the price of anarchy over pure Nash equilibria, nearly matching a lower bound
of $1.259$ for the case of three advertisers.   Further, we do not require that the system reaches
equilibrium, and give similarly low bounds also on the quality degradation
for any no-regret learning outcome. Our conclusion is that the number
of advertisers in the auction has almost no impact on the price of anarchy, and that
the efficiency of GSP is very robust with respect to the belief and rationality
assumptions imposed on the participants.
\end{abstract}



\newpage
\section{Introduction}

The sale of advertising space on the Internet, or AdAuctions, is the primary source of revenue for
many
providers of online services. According to a recent report \cite{emarketer},
$\$25.8$ billion dollars were spent in online advertisement in the US in
2010. The main part of this revenue comes from search advertisement, in which
search engines display ads alongside organic search results. The success of this
approach is due, in part, to the fact that providers can
tailor advertisements to the intentions of individual users, which can be inferred
from their search behavior.  A search engine,
for example, can choose to display ads that synergize well with a query being
searched.
However, such dynamic provision of content complicates the process of
selling ad space to potential advertisers.
Each search query generates a new set of advertising space to be
sold, each with its own properties determining the applicability of different
advertisements, and these ads must be placed near-instantaneously.

The now-standard mechanism for resolving online search advertisement
requires that each advertiser places a \emph{bid} that represents the maximum
she would be willing to pay if a user clicked her ad. These bids are then
resolved in an automated auction whenever ads are to be displayed.
By far the most popular bid-resolution method currently in use is the
Generalized Second Price (GSP) auction, a generalization of the well-known
Vickrey auction.  In the GSP auction, there are multiple ad ``slots'' of varying appeal
(e.g.\ slots at the top of the page are more effective).  In two seminal papers
Edelman et al.~\cite{edelman07sellingbillions} and Varian \cite{Varian06positionauctions}
propose a simple model of the GSP auction that we will also adopt in this paper.
They observe that truthtelling is not a dominant strategy under GSP, and GSP auctions do
not generally guarantee the most efficient outcome (i.e., the outcome that maximizes
social welfare). Nevertheless,
the use of GSP auctions has been extremely successful in practice.  This begs the
question: \emph{are there theoretical properties of the Generalized Second
Price auction that would explain its prevalence?} Edelman et al.~\cite{edelman07sellingbillions}
and Varian \cite{Varian06positionauctions}
provide a partial answer to this question by showing that,
in the full information setting, a GSP auction always has a Nash equilibrium that has same
allocation and payments as the VCG mechanism.
\cite{edelman07sellingbillions} and \cite{Varian06positionauctions}
give only informal arguments to justify the selection of envy-free equilibria.

We argue that the Generalized Second Price auction is best modeled as a
Bayesian game of partial information.
Modeling GSP as a full information game assumes that each auction is played
repeatedly with the same group of advertisers, and during such repeated play the bids stabilize.
The resulting stable set of bids is well modeled by a full information Nash equilibrium.
The analyses of Edelman et al.~\cite{edelman07sellingbillions} and
Varian \cite{Varian06positionauctions} provide important insight into the structure of
the GSP auction under this assumption.
However, the set and types of players can vary significantly between rounds of a
GSP auction.  Each query is unique, in the sense that it is defined not only by
the set of keywords invoked but also by the time the query was performed, the location
and history of the user, and many other factors. Search engines use complex machine learning
algorithms to select the ads, and more importantly to determine appropriate quality
scores (or factors) for each advertiser for a particular query, and then decide which advertiser
to display. This results in uncertainty both about the competing advertisers, and about
quality factors. We model this uncertainty by viewing the GSP auction as a Bayesian game, and ask:
what are the theoretical properties of the Generalized Second Price auction \emph{taking into account
the uncertainty that the advertisers face?}

\paragraph{Bounding the quality of outcomes: Price of Anarchy.}
To answer the question above, we offer
a quantitative understanding of the
inefficiencies that can arise in GSP auctions, using a metric known as the \emph{Price of
Anarchy}. We show that the welfare generated by the auction in \emph{any}
equilibrium of bidding behavior is at least a $\frac{1}{\eta}$-fraction
of the maximum achievable welfare (i.e., the welfare the auction could generate knowing
the player types and quality factors in advance). The value of $\eta$ measures the
robustness of an auction with
respect to strategic behavior: in the worst case, how much can strategic manipulation harm the
social welfare. The closer $\eta$ is to $1$, the more robust the auction is. An
auction that always generates efficient outcomes at equilibrium would have price of anarchy equal
to~$1$. We bound the inefficiency of the outcomes both in the Bayesian version of the game
as well as the full information game, and extend the analysis also for learning outcomes.

We develop a general technique for bounding the inefficiency of outcomes that
allows us to do this in the most general setting, even in Bayesian games with multiple,
correlated sources of uncertainty. Our framework of \emph{semi-smooth} games is an extension of
Roughgarden's \cite{roughgarden} smoothness framework, that allows dealing with correlated
distributions. Correlated distributions are an important feature of the GSP model, especially
when modeling quality factors, as the same facts affect clickability and hence the quality
factors for all advertisers. (For instance, an ad shown to a bot will not get a click independent
of the advertiser.)

For mechanisms that are not dominant strategy truthful,
like GSP auctions, price of anarchy analysis is a powerful
tool for quantifying the potential loss of efficiency at equilibrium. We conduct this
analysis both in a full information setting without uncertainty (in which the
price of anarchy is surprisingly small, indicating a loss of at most $22\%$ of the
welfare), but also in a setting with uncertainty and a very
general information structure, in which we prove that the price of anarchy is
still bounded by a small constant. This shows that while the GSP auction is
not guaranteed to be efficient, it is a reasonably good design, as remarkably, the
welfare loss of these auctions is bounded by a value that does not depend on the
number of players, the number of advertisements for sale, or the prior distributions
on player types.
In contrast, the variant of the Generalized Second Price auction that orders
advertisers by their bid ignoring quality factors, which has been
historically used by Yahoo!, results in a quality loss proportional to the range
of quality factors, while randomly assigning advertisers to slots can result in a
loss of efficiency proportional to the number of advertisers.


One feature of our results is that they hold for a variety of models regarding the
rationality and the beliefs of the players.
This robustness is particularly
important in large-scale auctions conducted over the Internet, where assumptions
of full information and/or perfect rationality of the participants are unreasonably strong.

\paragraph{The GSP auction and sources of uncertainty.}
By far the most popular auction method currently in use for search ads is the
Generalized Second Price (GSP) auction, a generalization of the well-known Vickrey auction.
The GSP auction is invoked every time a user queries a keyword of interest; it is a
repeated auction in which players repeatedly bid for ad slots.  However,
modeling equilibrium strategies in a repeated game of this nature is notoriously
difficult, and results in a game with a plethora of unnatural equilibria
due to the possibility of bids representing threats for
future rounds, optimal exploration of the bidding space, and so on.  A common
simplification used in the literature is to focus on auctions for a single
keyword, and to suppose that players will quickly learn each others' valuations and
reach a stationary equilibrium.  Under this assumption, the stationary
equilibrium would correspond naturally to a Nash equilibrium in the full
information, one-shot version of the GSP auction \cite{Edelman10repeated}.  It
has therefore become common practice to study pure, full information equilibria
of the one-shot game, as an approximation to expected behavior in the more
general repeated game \cite{edelman07sellingbillions, Varian06positionauctions,
PLT09_aaw}.

In reality, however, the set and types of players can vary significantly between
rounds of a GSP auction:  each query is unique, in
the sense that it is defined not only by the set of keywords invoked but also
the time the query was performed, the location and history of the user, and many
other factors.  This \emph{context} is taken into account by an underlying
\emph{ad allocation algorithm}, which is controlled by the search engine. The ad
allocation algorithm not only selects which advertisers will participate in an
auction instance, but also assigns a \emph{quality factor} to each advertiser. As a first
approximation we can think of the quality factor as a
score that measures how likely that participant's ad will be clicked for that
query. These quality factors are then used to scale the bids of the advertisers.
These scaled bids are known as \emph{effective bids}, which can be viewed as bids
derived from a similarly-modified  \emph{effective type}. Under our assumption that
quality factors measure clickability, the effective type of an advertiser is the
expected valuation of displaying the ad (valuation of the ad times its likelihood of getting a click).
The effective bid and effective type of a player are therefore random
variables, which can be thought of as the original valuations multiplied by
quality scores computed exogenously by the search engine.
Athey and Nekipelov \cite{AN10} point out that the uncertainty in quality
factors produces qualitative changes in the structure of the game. Thus, even if
players converge to a stationary bidding pattern,
 the resulting equilibrium cannot be described as the
outcome of a full information game.

We model the uncertainty about the effective types of advertisers as
a Bayesian, partial information game. That is, the inherent uncertainty
due to context and the ad allocation algorithm can be captured via prior
distributions over effective types, even when the true types of all potential
competitors are fully known.  The appropriate equilibrium notion is then
the Bayes-Nash equilibrium with respect to these distributions. Our model allows
arbitrary correlations between the types and quality factors. The uncertainty of
ad quality and allocation mostly comes from the query context, and hence is best
modeled by correlated distributions of types and ad quality. Search engines use complex
machine learning algorithms to compute quality factors based on all available
information about the context, whose outcome is hard to predict for the advertisers.
Search engines share distributional information about quality factors with
advertisers. We model this by assuming that the advertisers are aware of the
distribution of quality factors. Further, we also assume that
the quality factors computed by the search engine correspond exactly to the clickability of the
ad.

Summarizing, there are two main sources of uncertainty: the first is about the
quality factors that the search engine attributes to each
advertiser and the second is about the valuations (types) of
the players.
These sources are different in nature: each advertiser has knowledge of (and can
condition her behavior on) her own type, whereas quality factors are fully
exogenous and are only revealed ex post.



\paragraph{Asymmetric information.} There are
different types of players in advertising markets, which may have differing levels of
information about their competitors.  We assume all players know their own valuations
correctly, but some smaller players (such as individual
advertisers) might be clueless about the valuations of the other players and
expected behavior of quality scores, while others (say bidding agencies or
large companies with web advertising departments) may have a much better
understanding of how individual rounds of the auction will proceed.
Even among
this latter group, different advertisers may have access to different information.  We
can model such information asymmetries by giving each player access to an
arbitrary player-specific signal that can carry information about the
effective types of the auction participants. Our bounds
on social efficiency in the Bayesian model hold in settings with such
asymmetry in information.

\paragraph{Learning players.}
So far we have considered equilibria of the auction game.
Analyzing equilibria makes the strong assumption that players reach equilibrium
play. Learning outcomes provide a very appealing generalization.
A now standard model considers a repeated version of the game, and assumes that players
employ strategies that give them
vanishingly small \emph{regret} over time.   Roughly
speaking, such a model assumes that players observe the bidding patterns of
others and modify their own bids in such a way that their long-term performance
is at least as good as a single optimal strategy chosen in hindsight. Notice
that if all players employ the same (possibly randomized) strategy in each round, the
resulting stable strategies form a Nash equilibrium. Therefore, the no-regret
assumption of repeated play is a generalization of the notion of Nash equilibrium.
Further, there are many simple bidding strategies that yield vanishing regret over time,
as discussed below. The no-regret assumption does not require that players follow one of
 these algorithms; in fact, good play can result in better utility than simply no-regret,
e.g., if the player can anticipate the behavior of other players. Rather, the assumption
models a natural rationality: if there is a consistently good strategy players will
attempt to learn this over time, and do at least as well (or better) as this good
fixed strategy. In this sense, the no-regret assumption aims to capture the intuition
that players attempt to learn beneficial bidding strategies over time, while also
providing a generalization of Nash equilibrium play. We view the existence of simple
learning algorithms as supporting this assumption. If all players have no-regret
this will cause the empirical distribution of the bids to converge to a coarse
correlated equilibrium of the game, a slight generalization of the well-known correlated equilibrium.

We therefore assume that the players use algorithms to learn how to best bid given
their valuation and signal, and {achieve vanishing regret} over time.
In other words, for each possible valuation and signal, repeated auctions allow players to learn how to
best bid taking into account the varying bids of other players, and the uncertainty about
quality factors, other players' valuations, and bidding strategies. We will consider the
quality degradation of the average social outcome when all players employ strategies
with small regret. Blum et al.~\cite{BHLR-08} introduced the  term  \emph{Price
of Total Anarchy} for this analog of the price of anarchy.


\paragraph{Approximate rationality.} One of the fundamental assumptions in
auction analysis is that all players are perfectly rational utility
optimizers.  However, in reality (and especially in large online settings), it
is natural to assume that some fraction of the players participating in an
advertising auction might have unsophisticated bidding strategies.  In fact,
some players may not even play at equilibrium in the single-shot approximation of
the GSP auction, or may only be able to find strategies that are approximately
utility-maximizing.
We discuss the robustness of our bounds to the presence of players bidding with
limited (or no) rationality.
As we shall see, the GSP auction has the property that its social welfare
guarantees degrade continuously when our assumptions about the rationality of the players
are relaxed.

\subsection{Our results}
We present the following results.
\begin{itemize}
\item Our main result is a bound on the Bayesian
price of anarchy for the GSP auction.  Specifically, we show that the price of anarchy
is at most $2.927$, meaning that the social welfare in any
Bayes-Nash equilibrium is at least $1/2.927$ of the optimal
social welfare. Notice that this is an unconditional bound, as we make no assumptions on
the distribution on valuation profiles and quality factors (it can, for example,
be correlated) or on the number of players or slots. {In the main part of the
paper, we prove weaker bounds for both the full information and the Bayesian game, and
only sketch the stronger bounds. We believe that the weaker bounds are  interesting
{in} their own right, and show the main techniques of the paper {in a
way that is easier to read}. We defer the details of the stronger bounds to the Appendix.}

Perhaps just as important as the bound, however, is the
straightforward and robust nature of the GSP auction.  In particular, our results
extend to provide the same welfare guarantees for outcomes of no-regret learning:
the average social welfare when players play repeatedly in order to minimize total
regret, in a Bayesian setting, is within a $1/2.927$ factor of the optimal social welfare.
{In fact, some of our bounds for learning outcomes require only that the players
have no regret for a particular natural strategy of shading their bids.
The bounds continue} to hold even if players have asymmetric access to
distributional information, in the form of exogenously provided signals. It
also degrades continuously in the presence of approximately rational players or a
small fraction of irrational players {as explained in Appendix \ref{appendix:irrational}}.
{The results also extend to the case
when possible bid space is discretized (i.e., players need to bid in integer number
of pennies). This case is interesting both as in practice bids do come from such a
discrete space, and also as in the discrete case, { the existence of Nash
equilibria is guaranteed}. In fact, using a result of Athey \cite{Athey} and Reny \cite{Reny},
in the discrete case if player types and quality factors are drawn independently,
one can also show that {the existence of pure strategy equilibria that are
monotone in the types is guaranteed.}}

\item We achieve {the bounds on the solution quality by identifying a
property that encapsulates some of the insight.}
Roughgarden \cite{roughgarden} identified a class of games that he termed \emph{smooth}
games, defined via a similar property that is used to bound the price of anarchy.
We identify a {stronger} 
property, semi-smoothness, that is
satisfied by the GSP auction, and is strong enough to also imply price of
anarchy bounds {even in the Bayesian setting with arbitrarily correlated types}.

\item We provide improved results for the case where there is no uncertainty, which
is the traditional setting studied in \cite{edelman07sellingbillions,
Varian06positionauctions}. If valuations and quality factors are fixed, we prove
that the social welfare in any \emph{pure Nash equilibrium} is within a factor
of $1.282$ of the optimal one and show that this bound is essentially tight by
providing a lower bound of $1.259$. {Also, we show a bound of $2.310$ for
coarse correlated equilibria; as discussed above, this implies the same bound on
the social welfare for learning outcomes when players with fixed (effective) types
minimize their regret in a repeated auction.}
This bound of $2.310$ holds for mixed Nash equilibria as well.
\end{itemize}

\subsection{Related work}
Due to their central role in Internet monetization, sponsored search auctions have
received considerable attention in the past years. From the optimization
perspective, they were first considered by Mehta et al.~\cite{jacmMehtaSVV07}. A
classical game-theoretical modeling of sponsored search auctions was proposed simultaneously
by Edelman et al.~\cite{edelman07sellingbillions} and Varian
\cite{Varian06positionauctions}. See the surveys of Lahaie et al.~\cite{algorithmgame}
and Maille et al.~\cite{maille} for an overview of
subsequent developments.

The model we adopt follows
\cite{edelman07sellingbillions,Varian06positionauctions}. In those two seminal
papers, the authors notice that even though truthtelling is not a dominant
strategy under GSP, the full information game always has a Nash equilibrium that has same allocation
and payments as the VCG mechanism. They focus
on a subclass of Nash equilibria which is called \emph{envy-free equilibria} in
\cite{edelman07sellingbillions} and \emph{symmetric equilibria} in
\cite{Varian06positionauctions}. They show that such equilibria always exist and are
always efficient. In this class, an advertiser
would not be better off after switching bids with the advertiser just above
her.
Note that this is a stronger requirement than in Nash equilibria, which are defined
considering only unilateral deviations by the advertisers, and if an advertiser
unilaterally switches to a slot with higher click-through-rate, she pays more than the advertiser at that slot paid.
In \cite{edelman07sellingbillions,
Edelman10repeated,Varian06positionauctions}, informal arguments are presented
to justify the selection of envy-free equilibria, but no formal
game-theoretical analysis is done. We believe it is an important question to go
beyond this and prove efficiency guarantees for all Nash equilibria.
Lahaie \cite{lahaie} also considers the
problem of bounding the social welfare obtained at equilibrium, but restricts
attention to the special case that click-through-rates decay
exponentially along the slots with a factor of $\frac{1}{\delta}$.  Under this
assumption, Lahaie proves a price of anarchy of $\min\{\frac{1}{\delta}, 1 -
\frac{1}{\delta} \}$.

Gomes and Sweeney \cite{gomes09} study the GSP auction in
the Bayesian setting, where player types are drawn from independent and identical
distributions (without considering the uncertainty due to quality factors).
They show that, unlike the full information case, there
may not exist symmetric or socially optimal equilibria in this model, and obtain
sufficient conditions on click-through-rates that guarantee the existence of a
symmetric and efficient equilibrium. Athey and Nekipelov \cite{AN10} study the
effect of uncertainty of quality factors both from a theoretical and an empirical
perspective.

The study of price of anarchy for non-truthful auction mechanisms (especially in
the Bayesian setting) was initiated by
Christodoulou et al.~\cite{christodoulou} and developed in
Lucier and Borodin \cite{borodin_lucier}, Lucier \cite{lucier_ics}, and most
recently in the work of Bhawalkar and
Roughgarden \cite{roughgarden_bhawalkar}. To the best of our knowledge, the current
paper is the first one in which the price of anarchy bounds hold when
player valuations are drawn from a correlated distribution. In truthful
mechanism design, the study of correlated valuations has a long history -- see
Cremer and
McLean \cite{cremer-mclean} for an early reference.

The study of regret-minimization goes back to the work of Hannan on repeated
two-player games \cite{H-57}. Since then, a number of simple algorithms
(to be thought of as adaptive procedures) that guarantee no-regret have been proposed
in the literature. Initial work in this area focused on the stronger requirement of
finding simple adaptive procedures through which the play converges to the set of
correlated equilibria, requiring that players have a stronger form {of} no-regret
that is called no internal regret (see the survey by Blum and Mansour \cite{BlumMansour}
for a discussion of such procedures and a comparison). Foster and Vohra \cite{FosterVohra1997}
obtained such a procedure, and Fudenberg and Levine \cite{FudenbergLevine1999}
presented a different one. Hart and Mas-Collel's regret matching strategy \cite{hart00}
or the multiplicative weight updating strategy of
\cite{weighted-maj} (see also \cite{Arora-etal}) are two procedures that become especially
simple when used to guarantee only no-regret (as opposed to no internal regret).
{These classical learning algorithms assume that players learn
outcomes and strategies
of all participants in each round, but have also been extended to
situations}  {where in each round, a
player observes only her own outcome, or even realizations of her outcome in case it is
randomized. We refer to Auer et al. \cite{AU95} for a detailed discussion
on this matter}.

Adaptive procedures that guarantee no-regret define a play that converges to the set of coarse correlated equilibria.
Blum et al.~\cite{BHLR-08} apply regret-minimization to the study of inefficiency in
repeated games, coining the term ``price of total anarchy'' for the worst-case
ratio between the optimal objective value and the average objective value when
players minimize regret.

Roughgarden \cite{roughgarden} identifies a class of games that he terms \emph{smooth}
games where the price of anarchy and price of total anarchy are identical.
{See also  \cite{NadavRoughgarden2010} and \cite{RoughgardenSchoppmann2011}, for
subsequent refinements. Since the initial conference versions of our Bayesian bound of
\cite{LPL11} and \cite{CKKK11}, Roughgarden \cite{RoughgardenEC12} and independently
Syrgkanis \cite{Syrgkanis12} show that the bounds proved via smoothness also extend to the Bayesian price of anarchy assuming a variant of the smoothness assumption (called universal
smoothness in \cite{Syrgkanis12}) if player types are drawn from independent
distributions. See \cite{SyrgkanisTardos13} for such an extension theorem without
the stronger assumption.
In this paper we isolate a stronger} property related to smoothness that encapsulates
many of the insights that drive our bounds {and allows us to extend our bounds
for the Bayesian price of anarchy with correlated distributions}.

Some of the results in this paper appeared in preliminary conference versions.
Paes Leme and Tardos \cite{PLT10} study equilibria of GSP auctions and give upper bounds on the price of
anarchy in pure, mixed, and Bayesian strategies; achieving bounds of $1.618$,
$4$, and $8$, respectively.
Lucier and Paes Leme \cite{LPL11} and Caragiannis et al.~\cite{CKKK11} improve
these bounds {to $3.16$ and $3.037$ respectively for Bayesian games,
and $1.282$ and $2.31$ for {pure} Nash and learning outcomes for full
information games (as well as mixed Nash equilibria),} and extend them to apply to equilibria with correlated valuations
and learning outcomes.  {Here we further improve
the bounds, present and also improve the proofs, and extend the results to games with
uncertainty about quality factors in addition to player types. }

\section{Model and Equilibrium Concepts}

We consider an auction with $n$ advertisers and $n$ slots\footnote{We note that we can handle
unequal numbers of slots and advertisers by adding
virtual slots with click-through-rate zero or virtual advertisers with zero valuation per
click.}. Each advertiser $i$ has a private
type $v_i$, representing her valuation per click received.  The sequence
$\vals = (\val_1, \dotsc, \val_n)$ is referred to as the \emph{type profile} (or \emph{valuation profile}).
We will write $\valsmi$ for $\vals$ excluding the $i$th entry, so that
$\vals = (\vali, \valsmi)$.

An \emph{outcome} is an assignment of advertisers to slots.  An outcome can be
viewed as a permutation $\pi$ with $\pi(k)$ being the advertiser assigned to slot
$k$. The probability of a click depends on the slot as well as the advertiser
shown in the slot. We use the model of separable click probabilities. We assume
slots have associated \emph{click-through-rates} $\alpha_1 \geq \alpha_2 \geq
\hdots \geq \alpha_n$, and each advertiser $i$ has a \emph{quality factor}
$\gamma_i$ that reflects the clickability of the ad. When advertiser $i$ is
assigned to the $k$-th slot, she gets $\alpha_k \gamma_i$ clicks.

A mechanism for this auction
elicits a bid $\bidi \in \R_+ := [0,\infty)$ from each advertiser $i$, which is
interpreted
as a type declaration, and returns an assignment as well as a price
$\pricei$ per click for each advertiser.
If advertiser $i$ is assigned to slot
$j$ at a price of $\pricei$, her \emph{utility} is $\alpha_j \gamma_i(\vali -
\pricei)$, which is the number
of clicks received times profit per click. The \emph{social welfare}
of outcome $\pi$ is $SW(\pi,\vals,\gamma) = \sum_j \alpha_j
\gamma_{\pi(j)} \val_{\pi(j)}$, the total value
of the solution for the participants. The social welfare also depends on the click-through-rates $\alpha_j$, but throughout the paper we will assume they are fixed and common knowledge, and as a result we suppress them in the notation.
The optimal social welfare is $OPT(\vals,\gamma) = \max_\pi
SW(\pi,\vals,\gamma)$, the welfare
generated by the socially efficient outcome. Note that the efficient outcome
sorts advertisers by their \emph{effective values} $\gamma_i v_i$, and assigns
them to slots in this order. The effective value can be thought of as the
expected value of showing the ad in a slot with click-through-rate equal to $1$.

We focus on a particular mechanism, the Generalized Second Price auction,
which works as follows.
Given bid profile $\bids$, we define the \emph{effective bid} of advertiser $i$ to
be $\gamma_i b_i$, which is her bid modified by her quality factor, analogous to the effective value defined above.
The auction sets $\pi(k)$ to be
the advertiser with the $k$th highest effective bid (breaking ties arbitrarily).
That is, the GSP mechanism assigns slots with higher click-through-rate to advertisers with higher
effective bids.
Payments are then set according to critical value: the smallest bid that guarantees the advertiser the same slot. When advertiser $i$ is assigned to slot $k$ (that is, when $\pi(k)=i$), this critical value is defined as
$$\pricei = \frac{\gamma_{\pi(k+1)}}{\gamma_i}
\bid_{\pi(k+1)}$$
where
we take $\bid_{n+1} = 0$.  We will write
$\utili(\bids,\gamma)$ for the utility derived by advertiser $i$ from the GSP mechanism
when advertisers bid according to $\bids$:
$$u_i(\bids, \gamma) = \alpha_{\pi^{-1}(i)}  \gamma_i (v_i - p_i) =
\alpha_{\pi^{-1}(i)} [ \gamma_i v_i - \gamma_{\pi(\pi^{-1}(i)+1)}
b_{\pi(\pi^{-1}(i)+1)} ].
$$

Notice that $\pi$ is a function of $\bids, \gamma$ as well. In places where we
need to be more explicit, we will write $\pi(\bids,\gamma,j)$ to be the advertiser
assigned
to slot $j$ by GSP when quality factors are $\gamma$ and the advertisers bid according to $\bids$.  We will also
write $\sigma(\bids,\gamma,i)$ for the slot assigned to advertiser $i$, again when advertisers
bid according to $\bids$ and quality factors are $\gamma$.  In other words, $\sigma(\bids,\gamma,\cdot) = \pi^{-1}(\bids,\gamma,\cdot)$.
We write $\pi^i(\bidsmi,\gamma,j)$ to be the advertiser that would
be assigned to slot $j$ if advertiser $i$ did not participate in the auction.
When $\bids$ and $\gamma$ are clear from the context, we write $\pi(i)$ and $\sigma(i)$
instead of $\pi(\bids,\gamma,i)$ and $\sigma(\bids,\gamma,i)$.
We will also write $\nu(\vals,\gamma)$ for the optimal assignment of slots to advertisers
for valuation profile $\vals$, so that $\nu(\vals,\gamma,i)$ is the slot that would be allocated to
advertiser $i$ in the optimal assignment\footnote{We note that, since GSP makes the
optimal assignment for a given bid declaration, we actually have that $\nu(\vals,\gamma,i)$
and $\sigma(\vals,\gamma,i)$ are identically equal.  We define $\nu$ mainly for use when
emphasizing the distinction between an efficient assignment for a valuation profile and
the assignment that results from a given bid profile.}.

We will consider rational behavior under various models of the information available to the advertisers.
In general, the advertisers are engaged as players in a game defined by the auction mechanism; each of them aims
to select a bidding strategy that maximizes her utility. In the following, we use the terms advertiser and player interchangeably.
We group our models into full information and partial information ones.
In all models we assume that the values $\alpha_j$ are fixed and
commonly known to all players. In our full information settings,
we assume that the quality factors $\gamma_i$ as well as the valuation profile $\vals$
are also common knowledge.
In our Bayesian setting of partial information,
we assume that the profile of quality factors is unknown to all players,
and the type $\vali$ is private knowledge known only to player $i$, but they are randomly drawn from a commonly known joint distribution $(\dists,\mathbf{G})$ of quality factors and valuation profiles.
It will turn out that bidding more than one's true type (\emph{overbidding}) is a dominated strategy in
the mechanism we consider. So, we will focus on non-overbidding (or conservative) players; see Section \ref{subsec:nonoverbidding} for a discussion.

\subsection{Full information setting}

In the full information setting, the valuation profile $\vals$ and quality
factors $\gamma_i$ are fixed and common knowledge.   We will therefore tend to
drop dependencies on $\gamma$ from our notation when working in the full
information
setting.  In this setting,
a \emph{pure strategy} for player $i$ is a bid $b_i \in \R_+$. We say that the
bid profile $\bids$ is a \emph{(pure) Nash equilibrium} if there is no deviation
from which the player can profit, i.e., for all $b'_i \in \R_+$,
$$u_i(b_i, \bidsmi) \geq u_i(b'_i, \bidsmi).$$
It is known that a pure Nash equilibrium always exists in this setting
\cite{edelman07sellingbillions,Varian06positionauctions}. 
We can therefore define the \emph{(pure) Price of Anarchy} to be
$$\sup_{\vals, \bids \in NE}
\frac{OPT(\vals)}{SW(\pi(\bids),\vals)}$$
where NE is the set of pure Nash equilibria (assuming no overbidding; see Section \ref{subsec:nonoverbidding}).

Similarly, a \emph{mixed strategy} is
a randomized bid $b_i$, which is a distribution over possible bids. A mixed
Nash equilibrium is a profile of bid distributions $\bids$ such that for all $i$
and all alternative strategies $b_i'$,
\[
\E_{\bids} [ u_i(b_i, \bidsmi)) ] \geq \E_{\bids} [ u_i(b'_i, \bidsmi)) ].
\]
Note that, unlike more general solution concepts we will discuss in a while, the bid distributions of different players at a mixed Nash equilibrium are independent.
We define the \emph{(mixed) Price of Anarchy} to be the worst-case ratio between
optimal social welfare and expected social welfare in GSP across all
valuation profiles and all mixed Nash equilibria:
$$\sup_{\vals, \bids \in NE}
\frac{OPT(\vals)}{\E_{\bids}[SW(\pi(\bids),\vals)]}.$$

\subsection{Bayesian setting}

In the Bayesian setting of partial information, we suppose that the valuation profile and the quality factors are drawn from a publicly known (possibly
correlated) joint distribution $(\dists,\textbf{G})$.
A strategy for player $i$ is a (possibly randomized) mapping $b_i
: \R_+ \rightarrow \R_+$, mapping her type $v_i$ to a bid $b_i(v_i)$. Notice
that a player is \emph{not} able to condition her bid on the quality factors,
since they are only known to the search engine, and not to the advertisers.

We write $\bids(\vals) = (\bidi[1](\vali[1]), \dotsc, \bidi[n](\vali[n]))$ to denote the
profile of bids that results when $\bids$ is applied to type profile $\vals$.
We then say that strategy profile $\bids$ is a Bayes-Nash equilibrium for distributions
$\dists,\textbf{G}$ if, for all $i$, all $v_i$, and all alternative strategies $b_i'$,
\comment{
\[
\E_{\gamma,\valsmi \vert v_i}[\utili(\bidi(\vali), \bidsmi(\valsmi), \gamma) ]
\geq
\E_{\gamma,\valsmi \vert v_i}[\utili(\bid'_i(\vali), \bidsmi(\valsmi), \gamma)]
\]
}
\[
\E_{\valsmi,\gamma,\bids }[\utili(\bidi(\vali), \bidsmi(\valsmi), \gamma) \vert v_i ]
\geq
\E_{\valsmi,\gamma,\bids }[\utili(\bid'_i(\vali), \bidsmi(\valsmi), \gamma)
\vert v_i].
\]
That is, each player maximizes her expected utility by bidding in accordance with strategy $b_i(\cdot)$,
assuming that the other players bid in accordance with strategies $\bidsmi(\cdot)$,
where expectation is taken over the distribution of the other players' types
conditioned on $v_i$, any randomness in
their strategies, and the quality factors.
We define the \emph{Bayes-Nash Price of Anarchy} to be
$$\sup_{\dists, \mathbf{G}, \bids(\cdot) \in BNE} \frac{\E_{ \vals,\gamma}
[OPT(\vals,\gamma)]}{\E_{ \vals,\gamma,
\bids(\vals)}[SW(\pi(\bids(\vals),\gamma),\vals,\gamma)]}$$
where BNE is the set of all Bayes-Nash equilibria (again assuming no overbidding; see below).

\subsection{No overbidding}\label{subsec:nonoverbidding}

It is important to note that, in both the full information and Bayesian
settings, any bid $b_i > v_i$ is dominated
by the bid $b_i = v_i$ in the GSP auction. If by bidding $b_i > v_i$, the
next highest effective bid is greater than $\gamma_i v_i$, then the player gets
negative utility.
If on the other hand, the next highest effective bid is smaller or equal than
$\gamma_i v_i$, then
bidding $b_i = v_i$ would get the same slot and payment. Based on this, we make the
following assumption for the rest of the paper:\\

\noindent \textbf{Assumption:}  Players are \emph{conservative} and do not employ overbidding strategies in
GSP auctions. This means that for pure strategies $b_i \leq v_i$, for mixed strategies
$\prob(b_i > v_i) = 0$, and for Bayesian strategies $\prob(b_i(v_i) > v_i) = 0$
for all $v_i$.\\

We use this assumption to rule out unnatural equilibria in which
advertisers apply {certain} dominated strategies.  We remark that, in these
equilibria, the social welfare may be arbitrarily worse than the optimal.  It
is therefore necessary to exclude {such} dominated strategies in order to obtain
meaningful bounds on the price of anarchy.  We note, however, that this
phenomenon is not specific to the GSP auction: such degenerate equilibria
exist even in the Vickrey auction for a single good, where truthful
bidding is a weakly dominant strategy.  Since the Vickrey auction is a
special case of GSP auctions (where one slot
has $\alpha_1=1$, all other slots have $\alpha_i=0$ and all quality factors have $\gamma_i = 1$), this issue carries over to our setting. Consider the example of a single-item Vickrey
auction, where truthful bidding of $b_i=v_i$ is a weakly dominant strategy. Yet
with overbidding, there are equilibria where an arbitrary player
bids excessively high (and hence wins), while everyone else bids $0$. If the
player bidding high has a low {valuation}, this results in a high price of
anarchy.
Note, however, that this Nash equilibrium seems very artificial as it depends
crucially on the low valuation player using the dominated strategy of overbidding.
Indeed, such an advertiser is exposed to the risk of negative utility (if
some other advertiser submits a new bid between her valuation and bid)
without any benefit.  We therefore take the position that advertisers will
avoid such dominated strategies when participating in the GSP auction.

\subsection{Signals and information asymmetry}\label{subsec:assymmetric_informatation}

We define an extension of the setting above, incorporating a Bayesian version of
information asymmetry. In this model, each player's type consists of a
signal $s_i$ drawn from an arbitrary signal space $S$. The signal of player $i$
includes her valuation $v_i(s_i)$ and can contain other privately-gained insight
that refines the player $i$'s conditional distribution over the space of other players' types and quality factors. The signals and quality factors come from a publicly
known joint distribution $(\dists', \mathbf{G})$, which can be arbitrarily correlated.

In this model, a strategy is a bidding function that maps $s_i$, a signal, to a
distribution of possible bids. The bid profile $\bids$ is a Bayes-Nash equilibrium in the asymmetric information model if,
for all $i$ and all alternative bidding functions $\bidi'$,
\[
 \E_{ \textbf{s}_{-i},\gamma,\bids }[\utili(\bidi(s_i),
\bidsmi(\textbf{s}_{-i}),\gamma) \vert s_i]
\geq
\E_{\textbf{s}_{-i}, \gamma,\bids }[\utili(\bid'_i( s_i),
\bidsmi(\textbf{s}_{-i}),\gamma)\vert s_i]
 \]
In this model, the Price of Anarchy is defined as
$$\sup_{\dists', \mathbf{G}, \bids(\cdot) \in BNE} \frac{\E_{\signal, \gamma
}[OPT(\vals(\signal),\gamma)]}{\E_{\signal ,\gamma,
\bids(\signal)}[SW(\pi(\bids(\signal),\gamma),\vals(\signal),\gamma)]}$$
where BNE is the set of Bayes-Nash equilibria with respect to distribution $\dists'$ over signals, with no overbidding.

The presence of signals captures the notion that some advertisers might have a better potential to infer the other
advertisers' valuations than others, or may be endowed with privileged information. We
do note, however, that players do know their own valuations $v_i(s_i)$, and also are aware of the profile of bidding strategies $\bids(\cdot)$ and the distribution $\dists'$,
so that players can rationalize about the effects of signals upon the bidding behavior of their opponents.

\subsection{Repeated auctions and regret minimization}

We now consider the GSP auction in a repeated-game setting.  In this model, the
GSP auction is run $T \geq 1$ times.  We will distinguish between two variants
of this model: the full information model and the model with uncertainty.

\paragraph{Full information model.} Each round of the GSP auction occurs
with the same slots and players.  The valuation profile $\mathbf{v}$ of the
players and the quality factors do not change between rounds, but the players
are free to change their bids. We write $b_i^t$ for the bid of player $i$ on
round $t$. We refer to a $D = (\bids^1, \dotsc, \bids^T, \ldots)$ as an (infinite)
\emph{declaration sequence}. Given declaration sequence $D$, we will write $D^T$
to mean the prefix of $D$ of length $T$; that is, $D^T = (\bids^1, \dotsc, \bids^T)$.
Given a (finite or infinite) declaration sequence $D$, we will write $\pi(D)$
for the sequence of permutations generated by GSP on input $D$.  The average
social welfare generated by GSP on a finite input sequence $D^T$ of length $T$
is
$SW(\pi(D^T),\vals) = \frac{1}{T}\sum_{t=1}^T SW(\pi(\bids^t),\vals).$  The
average
social welfare generated by GSP on an infinite input sequence $D$ is then
defined to be $SW(\pi(D),\vals) = \liminf_{T \to \infty} SW(\pi(D^T),\vals)$.

The full range of equilibria in such a repeated game is very rich, so we
restrict ourselves to a particular non-equilibrium form of play that
nevertheless captures the intuition that players learn appropriate bidding
strategies over time, without necessitating convergence to a stationary
equilibrium.

We say that declaration sequence $D = (\bids^1, \dotsc, \bids^T, \ldots)$ \emph{minimizes
external regret} for player $i$ if, for any fixed declaration $b'_i$,
\[ \sum_{t \leq T} u_i(b_i^t,\bidsmi^t) \geq \sum_{t \leq T} u_i(b'_i,\bidsmi^t) +
R(T)\]
where $R(T)/T \rightarrow 0$ as $T$ grows large.  That is, as $T$
grows large, the utility of player $i$ in the limit is no worse than the utility of the optimal
fixed strategy in hindsight.  The \emph{Price of Total Anarchy} \cite{BHLR-08}
is the worst-case ratio between optimal social welfare and the average
social welfare obtained by GSP across all declaration sequences that minimize
external regret for all players.  That is, the price of total anarchy is
$$\sup_{\vals, D} \frac{OPT(\vals)}{SW(\pi(D),\vals)}$$
where the supremum is taken over (infinite) declaration sequences that minimize external regret for all players.

To this point we have not discussed \emph{how} the players achieve vanishing
regret, and indeed our results are agnostic to this process.
There are many known learning algorithms that guarantee vanishing regret as $T$
goes to infinity. These algorithms require only that a player observes
her realized payoff after each round.  That is, to facilitate learning it is
enough if players
see whether or not their ad was clicked and, if so, the price of the
click.  In particular,
they do not need to learn outcomes or bids of other players, nor even the actual slot their ad was placed in.  For an extensive
discussion on no-regret algorithms with limited feedback, we refer to Auer
et al. \cite{AU95}.  Further, it is known that the price of total anarchy is closely related to an equilibrium notion for the single-shot game known as coarse correlated equilibrium.  We discuss
this relationship further
in Section \ref{sec.learning-full}.

\comment{
Given an instance of the GSP auction in the full information model, a \emph{coarse correlated equilibrium} (CCE) is a randomized bid profile $\bids$ (i.e.\ a distribution over possible bidding profiles) such that for all $i$
and all possible bids $b_i'$,
\[
\E_{\bids} [ u_i(b_i, \bidsmi)) ] \geq \E_{\bids} [ u_i(b'_i, \bidsmi)) ].
\]
Note that the difference between a coarse correlated equilibrium and a mixed Nash equilibrium is that the strategies of the player can be arbitrarily correlated in a CCE.  The price of anarchy in coarsely correlated equilibria is defined to be
$$\sup_{\vals, \bids \in CCE}
\frac{OPT(\vals)}{\E_{\bids}[SW(\pi(\bids),\vals)]}.$$
It is known that, for any game, the price of anarchy in coarsely correlated strategies is precisely equal to the price of total anarchy \cite{BHLR-08}.  Thus, to analyze the performance of GSP when players play repeated strategies, it is sufficient to analyze coarsely correlated equilibria of GSP.
}

\paragraph{Learning with uncertainty.} Next we describe our model of learning in
repeated GSP auctions with uncertainty.  In this model, each round of the GSP auction
occurs with the same slots, but the valuation profile $\vals$ and quality factors
$\gamma$ are redrawn from $(\mathbf{F},\mathbf{G})$ on each round\footnote{In
fact, we can also think of the set of players as changing on each round: if
player $i$ is assigned type $0$ on a given round, this can be interpreted as
player $i$ not being present in that round.}.  These changes to ad
quality and types can be thought of as being due to the context of the search query that initiates
each auction instance, which can change between rounds.
{As before, learning requires only that players observe their own outcome each
round, and not the
results for other players.} {I.e., a player learns whether or not her ad was
clicked,
and if so the price per click, but does not necessarily
observe the outcomes or bids of other players nor the realization of $\gamma$.}
{We again refer to Auer
et al. \cite{AU95} for a
discussion of no-regret algorithms with limited feedback.}

Suppose that each player has a finite type space\footnote{For instance, one could
assume that valuations are bounded and multiples of some arbitrarily small
increment.}.  Let $\vals^t, \gamma^t$ be the type profile
and quality factors drawn at round $t$. Given a declaration sequence $D$ and
type $\tilde{v}_i$ for player $i$, we denote by $I(i,\tilde{v}_i)$
the subsequence of $D$ consisting of the set of rounds in which player $i$ has
type $\tilde{v}_i$, i.e., $I(i,\tilde{v}_i) = \{t; v_i^t = \tilde{v}_i\}$. Define $I^T(i,\tilde{v}_i)$ analogously with respect
to $D^T$. Given a sequence of type profiles and quality factors that
represent the realization of these random quantities over the rounds of the
auction, we say that player $i$ has vanishing regret in declaration sequence $D$
if player $i$ has vanishing regret (in the sense of the full information game) on
the subsequence $I(i,\tilde{v}_i)$ of $D$ for each possible type
$\tilde{v}_i$. Formally:
\[ \sum_{t \in I^T(i,\tilde{v}_i)} u_i(b_i^t,\bidsmi^t,\gamma^t) \geq
\sum_{t \in I^T(i,\tilde{v}_i)}  u_i(b'_i,\bidsmi^t,\gamma^t) +
R(\abs{I^T(i,\tilde{v}_i) })\]
for $R(T)/T \rightarrow 0$ as $T \rightarrow \infty$. Notice that since
$\vals^t$ is independently and identically distributed in each
round, we have $\abs{I^T(i,\tilde{v}_i) } \rightarrow \infty$ as $T \rightarrow \infty$. Now,
we can define the \emph{Price of Total Anarchy with uncertainty} as:
$$\sup_{\{\vals^t\}, \{\gamma^t\}, D} \limsup_{T\rightarrow \infty} \frac{
\sum_{t \leq T}
OPT(\vals^t, \gamma^t) }{
\sum_{t \leq T}
SW(\pi(\bids^t,\gamma^t), \vals^t, \gamma^t)}.$$

As in the full information setting, there is a relationship between regret
minimization under uncertainty and coarse correlated equilibria with
uncertainty.  {Note however, that the speed of learning now depends on
the time {needed for} the empirical distribution of iid samples $(\vals^t,
\gamma^t)$
to resemble the original distribution {and for learning algorithms to
guarantee low regret with high probability}. We discuss this in more detail in
Section \ref{subsec:learning_with_uncertainty}.}

\section{Semi-Smooth Games and the Price of Anarchy with Uncertainty}\label{sec:uncertainty}

Our main result is a bound on the price of anarchy for the Generalized
Second Price auction with uncertainty.  Recall that our model captures two types
of uncertainty: uncertainty for player types and uncertainty about quality
factors. Further, our result holds even in the presence of information asymmetry
in the form of personalized signals available to the players.\footnote{In the presence of additional signals, we can assume that signal $s$ also
encodes the valuation of the player, i.e., that player $i$'s valuation for a click when
she receives signal $s_i$ is $v_i(s_i)$, and in this case, signals and quality factors are
drawn from a known joint distribution $(\mathbf{F}',\mathbf{G})$. Our statement and proof carry over to this case with straightforward modifications.} For simplicity of
presentation, we focus on the setting where there are no signals and
player valuations and quality factors are drawn from a known joint distribution $(\dists,\mathbf{G})$.

\begin{theorem}
\label{thm:cbpoa_main}
The price of anarchy of the Generalized Second Price auction with uncertainty is at most
$2.927$. That is, for any fixed click-through-rates $\alpha_1, \hdots, \alpha_n$, any
joint distribution $(\dists,\mathbf{G})$ over valuation profiles and quality
factors, and any Bayes-Nash equilibrium $\bids$,
 $$\E_{\vals,\gamma, \bids}[SW(\pi(\bids,\gamma),\vals,\gamma)] \geq \frac{1}{2.927}
\E_{\vals,\gamma}[OPT(\vals,\gamma)].$$
\end{theorem}

\paragraph{Semi-smooth games and the price of anarchy.}
Our proof is based on an extension of a proof technique introduced by
Roughgarden \cite{roughgarden}, which he calls smoothness. We begin by reviewing
this notion briefly in the context of a general game.
Let $\ts$ denote the (fixed) player types in a game, and $\hs$ a pure strategy
profile for the players, and let $U_i(\ts,\hs)$ denote the utility of  player
$i$ with player types $\ts$, and strategy profile $\hs$.
Let $sw(\ts, \hs)$ denote the social welfare generated by strategy profile
$\hs$, and $sw^*(\ts)$ the maximum possible social welfare.
Roughgarden defines $(\lambda,\mu)$-smooth games as games where for all pairs
of
pure strategy profiles $\hs, \hs'$, and any (fixed) vector of types $\ts$, we have
$$\sum_i U_i(\ts, h'_i, \hs_{-i}) \geq \lambda \cdot sw(\ts, \hs') - \mu \cdot sw(\ts, \hs).$$
Roughly speaking, smoothness captures the property that if strategy profile $\hs'$
results in a significantly larger social welfare than another strategy profile
$\hs$, then a large part of this gap in welfare is captured by the marginal increases in the
utility of each individual player when unilaterally switching her strategy from
$h_i$ to $h'_i$.

{It is not hard to see that GSP does not satisfy this definition for all
pairs of strategy profiles $\hs, \hs'$.  However, we argue
that GSP is smooth with respect to a
\emph{particular} (possibly randomized) strategy profile $\hs'$, as defined by Nadav and Roughgarden \cite{NadavRoughgarden2010}, that can be used by players unilaterally to improve
the efficiency of GSP whenever its allocation resulting from a
pure strategy profile $\hs$ is highly inefficient. Note that unlike
\cite{NadavRoughgarden2010} we require improvement relative to the social
optimum $sw^*(\ts)$ and not relative to $sw(\ts,\hs')$, i.e., {we} will not
assume that  $sw(\ts,\hs')$ is (close to) the maximum $sw^*(\ts)$.
Further, we will show that there exists such a strategy profile $\hs'$ where
the strategy $\hs'_i$ of a player depends only on the type of the player. We call games that satisfy this stronger requirement semi-smooth.}

\begin{defn}[semi-smooth games]\label{defn:semi-smooth-games}
We say that a game is \textbf{$(\lambda,\mu)$-semi-smooth} if for each player $i$ there exists some
(possibly randomized) strategy $h_i'(\cdot)$ (depending only on the type of the player) such that,
$$\sum_i \E_{h'_i(t_i)} [ U_i(\ts, h'_i(t_i), \hs_{-i})]  \geq \lambda \cdot sw^*(\ts)
- \mu \cdot sw(\ts,\hs),$$
for every pure strategy profile $\hs$ and every (fixed) type vector $\ts$. The expectation is taken over the random bits of $h_i'(t_i)$.
\end{defn}

Analogous to  Roughgarden's \cite{roughgarden} proof {(see also Nadav and
Roughgarden \cite{NadavRoughgarden2010}),} semi-smoothness also immediately
implies a bound on the price of anarchy with uncertainty {even when the
types are arbitrarily correlated}.

\begin{lemma}\label{lem.smoothness-to-bpoa}
If a game is $(\lambda,\mu)$-semi-smooth and its social welfare is at least the sum of the players' utilities, then the price of anarchy with uncertainty (and information asymmetries) is at most $(\mu+1)/\lambda$.
\end{lemma}
\begin{proof}
Consider a game in the Bayesian setting where player types are drawn from a joint probability distribution and let
$\hs$ be a Bayes-Nash equilibrium for this game. By the definition of the Bayes-Nash equilibrium, we have that
$\E_{\tsmi,\hs} [U_i(\ts, \hs) \vert t_i] \geq \E_{\tsmi,\hs}[U_i(\ts, h'_i(t_i), \hs_{-i}) \vert t_i]$
for every value the random variable $h'_i(t_i)$ may take. Hence,
$\E_{\tsmi,\hs} [U_i(\ts, \hs) \vert t_i] \geq \E_{\tsmi,\hs}\E_{h'_i(t_i)}[U_i(\ts, h'_i(t_i), \hs_{-i}) \vert t_i]$.
Now taking expectation over $t_i$, we get
$\E_{\ts,\hs} [U_i(\ts, \hs)] \geq \E_{\ts,\hs}\E_{h'_i(t_i)}[U_i(\ts, h'_i(t_i), \hs_{-i})]$.
By summing over all players, and using the fact that the social welfare is at least the sum of the players' utilities, as well as the semi-smoothness property, we have
\begin{eqnarray*}
\E_{\ts,\hs}[sw(\ts,\hs)] &\geq & \E_{\ts,\hs}[\sum_i{U_i(\ts,\hs)}]\\
&\geq & \E_{\ts,\hs}[\sum_i{\E_{h'_i(t_i)}[U_i(\ts,h'_i(t_i),\hs_{-i})]}]\\
&\geq & \E_{\ts,\hs}[\lambda \cdot sw^*(\ts) - \mu \cdot sw(\ts,\hs)]\\
& = & \lambda \E_\ts[sw^*(\ts)] - \mu \E_{\ts,\hs}[sw(\ts,\hs)].
\end{eqnarray*}
Note that the third inequality follows by applying the semi-smoothness property
for every fixed type vector and every pure strategy profile that are
simultaneous outcomes of the random vectors $\ts$ and $\hs$. The last inequality
implies $\E_{\ts}[sw^*(\ts)] \leq \frac{\mu+1}{\lambda}
\E_{\ts,\hs}[sw(\ts,\hs)]$, as claimed.
\end{proof}

We remark that the proof holds without significant changes if we add information
asymmetries in the game, i.e., if we assume that each player gets signals that
reveal her type and refine her knowledge on the probability distributions of the
types of the other players. {The only change required is to define an
extended type for each player, consisting of the player's original type composed with
that player's signal, and use it in place of the original type.}

{A particular strength} of Lemma \ref{lem.smoothness-to-bpoa} lies in the fact
that it can provide bounds on the efficiency loss for Bayesian games
{even with correlated types} (and, as we will see later in Section
\ref{sec.learning}, under even more general equilibrium concepts) by examining
substantially more restricted settings. In the context of GSP auction games, it
allows us to focus on identifying a (possibly randomized) deviating bid strategy
for each player (i.e., a bid $b'_i$ for each player $i$) so that the
semi-smoothness inequality holds for every fixed valuation vector $\vals$ and
pure bidding profile $\bids$.
By Lemma \ref{lem.smoothness-to-bpoa}, this then immediately implies
a bound on the price of anarchy of GSP auction games with uncertainty and information asymmetries.

\begin{remark}
{It is maybe easier to interpret a deterministic special case of Lemma
\ref{lem.smoothness-to-bpoa} where we require that Definition
\ref{defn:semi-smooth-games} holds for deterministic bids $h'_i(\cdot)$. As a
warmup {in analyzing} the GSP auction, we will show {in} Claim
\ref{smoothness-claim} that the bids $b'_i=\frac{1}{2}v_i$ can serve to prove
that the GSP auction is $(1/2,1)$-smooth, and hence has a price of anarchy of at
most 4. This bound of 4 on the price of anarchy {shows} that if the
social welfare is less than a quarter of the maximum possible, there is a player
$i$ who can deviate to $b'_i = \frac{1}{2} v_i$, a natural shading of
{her} bid, and improve {her utility}.}

{To  improve the bound we will consider a deviating bid $b_i = \theta \cdot
v_i$ for other constants $\theta \in (0,1)$. In fact, we will need to {consider}
a
random $\theta$ rather than a constant one.
There are two ways to understand such a random bid: a direct conclusion is that sampling $\theta$  according to the prescribed distribution produces a
good deviation in expectation, whenever welfare is low. But maybe a more natural interpretation is through the lenses of the \emph{probabilistic method}, used in combinatorics to
show that a certain object exists without finding it explicitly. If {there
exists} a randomized deviation $b_i = \theta \cdot v_i$ that improves
{player $i$'s utility}, this implies that there exists a deterministic
bid $\theta v_i$ that improves {player $i$'s utility}.}

{The randomization on selecting the bid $h'_i(\cdot)$ in Definition
\ref{defn:semi-smooth-games} gives us more flexibility to prove the
semi-smoothness inequality with good parameters $\lambda$ and $\mu$ by defining
appropriately the density function of $h'_i$'s.}
\end{remark}

\paragraph{Price of anarchy of GSP auctions.}
First note that, technically speaking, the GSP auction does not immediately fit
into the framework of semi-smoothness: advertiser payoffs depend on random
quality factors which may be correlated with the type profile.  However, this
notational technicality is easily addressed by expressing advertiser utilities
in expectation over quality scores. 
That is, expressing utilities in the GSP auction in the notation of general games, we have $U_i(\vals,\bids) = \E_\gamma[u_i(\bids,\gamma) \vert \vals]$. Since quality factors affect the social welfare as well, we have $sw^*(\vals) = \E_\gamma[OPT(\vals,\gamma) \vert \vals]$ and $sw(\vals,\bids) =
\E_\gamma[SW(\pi(\bids,\gamma),\vals,\gamma) \vert \vals]$.

We are ready to prove that GSP auction games are semi-smooth.
We start by presenting a slightly weaker version of Theorem \ref{thm:cbpoa_main},
where we prove a bound of $3.164$. Then we sketch the proof of the improved bound of $2.927$, which is more technically involved. Details of the proof can be found in Appendix \ref{appendix:bayes}.


\begin{lemma}\label{smoothness-lemma}
The GSP auction game is $(1-\frac{1}{e},1)$-semi-smooth.
\end{lemma}

\begin{proof}
We begin by rewriting the definition of semi-smoothness in the notation of GSP auctions.  The GSP auction game is $(1-\frac{1}{e},1)$-semi-smooth if and only if, for each valuation profile $\vals$, there exists a (possibly randomized) bid profile $\bids'$ (with $\bid'_i$ depending only on the valuation of player $i$) such that, for every bid profile $\bids$,
\begin{equation}
\label{eq.semismooth.expect}
\sum_i{\E_{\gamma,b'_i}[u_i(b'_i,\bidsmi,\gamma) \vert \vals]} \geq \left( 1-\frac{1}{e} \right)\E_\gamma[OPT(\vals,\gamma) \vert \vals] - \E_\gamma[SW(\pi(\bids,\gamma), \vals, \gamma) \vert \vals].
\end{equation}
We will actually establish the stronger property that this inequality holds for
\emph{all} $\gamma$, and not {only} in expectation.
\begin{equation}
\label{eq.semismooth.noexpect}
\sum_i{\E_{b'_i}[u_i(b'_i,\bidsmi,\gamma)]} \geq \left( 1-\frac{1}{e} \right)OPT(\vals,\gamma)  - SW(\pi(\bids,\gamma), \vals, \gamma).
\end{equation}
The desired inequality \eqref{eq.semismooth.expect} will then follow by taking \eqref{eq.semismooth.noexpect} in expectation over the choice of $\gamma$ (whose distribution may depend on the valuation profile $\vals$).

Before establishing inequality \eqref{eq.semismooth.noexpect}, we will prove the even weaker statement that the GSP auction game is $(1/2, 1)$-semi-smooth (which implies a bound of $4$ on
the price of anarchy with uncertainty).
\begin{claim}\label{smoothness-claim}
The GSP auction game is $(1/2,1)$-semi-smooth.
\end{claim}
\begin{proof}
Choose a vector $\vals$ of fixed valuations, a pure bidding
profile $\bids$, and quality factors $\gamma$.
Consider a  (deterministic) deviating bid $b'_i = v_i/2$ for each player $i$. We
distinguish between two cases (recalling that $\nu(i)$ is the slot assigned to player $i$
in the efficient allocation given $\vals$ and $\gamma$):

\begin{itemize}
\item If by bidding $b'_i$ player $i$ gets slot $\nu(i)$ or better, then
$u_i(b'_i, \bidsmi, \gamma )\ge \alpha_{\nu(i)} \gamma_i v_i/2$, as the payment
$p_i$ cannot exceed her effective bid.
\item If by bidding $b'_i$ player $i$ gets a slot lower than $\nu(i)$, then the
effective value of the player $\pi(\nu(i))$ in slot $\nu(i)$ is at least $\gamma_i
v_i/2$, as we assume no overbidding.
\end{itemize}
We conclude that, in either case,
$$
u_i(b'_i, \bidsmi, \gamma )\ge \alpha_{\nu(i)} \gamma_i v_i/2 -\alpha_{\nu(i)}
\gamma_{\pi(\nu(i))} v_{\pi(\nu(i))}.$$
Summing over all players, and noticing that $\sum_i \alpha_i
\gamma_{\pi(i)} v_{\pi(i)}=SW(\pi(\bids,\gamma), \vals, \gamma)$, while $\sum_i
\alpha_{\nu(i)} \gamma_i v_i=OPT(\vals, \gamma)$, we arrive at the claimed bound that
the GSP auction game is $(1/2, 1)$-semi-smooth:
$$\sum_i{u_i(b'_i,\bidsmi,\gamma)} \geq
\frac{1}{2}OPT(\vals,\gamma) - SW(\pi(\bids,\gamma), \vals,\gamma).$$
\end{proof}

Notice that the proof uses a single Nash inequality: that no player $i$ would be
better off changing her bid to $b'_i = v_i/2$, bidding half her {valuation}, a
natural shading of her {valuation}. As we will see in
Section~\ref{sec.learning}, the bound will also apply to learning outcomes under
the same assumption of not regretting this single alternative.

Now we return to proving the $(1-\frac{1}{e},1)$ semi-smoothness. To do this, consider a randomized bid $\bids'$,
rather than the deterministic bid of $b'_i=v_i/2$ considered above, that offers a more sophisticated bid-shading strategy. We consider a random strategy where player $i$ shades her bid randomly to a value in the interval $[0,v_i(1-\frac{1}{e})]$, where bid $b'_i$
is a random variable with density $f(y) = \frac{1}{v_i-y}$ for $y \in [0,
v_i(1-\frac{1}{e})]$ and $f(y) = 0$ otherwise. We will show that
\begin{equation}
\label{eq.smooth.1}
\E_{b'_i} [u_i(b'_i, \bids_{-i}, \gamma)] \geq \left( 1-\frac{1}{e} \right)
\alpha_{\nu(i)}
\gamma_i v_i -
\alpha_{\nu(i)} \gamma_{\pi( \nu(i))} b_{\pi(\nu(i))}.
\end{equation}
Like in the proof of Claim \ref{smoothness-claim}, by summing expression \eqref{eq.smooth.1} for all $i$ and using the fact
that $b_{\pi(i)} \leq v_{\pi(i)}$ by the
non-overbidding assumption, we obtain that the game is
$(1-\frac{1}{e},1)$-semi-smooth.

It remains to derive equation \eqref{eq.smooth.1}. 
We have that
\begin{eqnarray*}
\E_{b'_i} [u_i(b'_i, \bidsmi, \gamma)] & \geq & \E_{b'_i} [ \alpha_{\nu(i)}
\gamma_i ( v_i
- b'_i) \one{\gamma_i b'_i \geq \gamma_{\pi(\nu(i))}
b_{\pi(\nu(i))}} ] \\
& = & \int_0^{v_i( 1-\frac{1}{e} )} \alpha_{\nu(i)} \gamma_i (v_i
- y) \one{\gamma_i y \geq \gamma_{\pi(\nu(i))}
b_{\pi(\nu(i))}} \frac{1}{v_i-y} dy \\
& = & \alpha_{\nu(i)} \gamma_i \left[ v_i\left( 1-\frac{1}{e} \right) -
\frac{\gamma_{\pi(\nu(i))}}{\gamma_i} b_{\pi(\nu(i))} \right]^+\\
& \geq & \left( 1-\frac{1}{e} \right)
\alpha_{\nu(i)}
\gamma_i v_i - \alpha_{\nu(i)} \gamma_{\pi(\nu(i))} b_{\pi(\nu(i))}
\end{eqnarray*}
which implies {\eqref{eq.smooth.1},}
completing the proof of Lemma \ref{smoothness-lemma}.
\end{proof}

Combining Lemmas \ref{lem.smoothness-to-bpoa} and \ref{smoothness-lemma}, we get the claimed bound on the price of anarchy.
\begin{theorem}
\label{thm:cbpoa}
The price of anarchy of the Generalized Second Price auction with uncertainty (and with information asymmetries) is at most
$2(1-1/e)^{-1} \approx 3.164$.
\end{theorem}

{To prove Theorem \ref{thm:cbpoa_main} we will need to extend semi-smoothness
specially tailored to the GSP auction game. In addition to working more
carefully on optimizing constants, we want to highlight two ideas: First, note
{that the} payment of the player in slot $i$ is $\alpha_i
\gamma_{\pi(i+1)} b_{\pi(i+1)}$, and this payment also contributes to the social
welfare. Using that $\alpha$ is monotone decreasing, we {can add} a term
$\alpha_i \gamma_{\pi(i)} b_{\pi(i)}$ to social welfare (in addition to the
advertisers' utility) for all slots
 except the top one. Second, the player in the top slot can obtain a stronger
bound on her utility by considering the deviation $b'_1=v_1$. (For all other
players bidding too close to $v_i$ endangers getting a higher slot at too high a
price, but the top player {does not} face this danger.) We will use
{these ideas} in Section \ref{sec.learning} to improve our bound on the
learning outcomes for full information games. The details of the more
complicated improved bound for the Bayesian case are found in Appendix
\ref{appendix:bayes}.}

\paragraph{Discretization of the bidding space.}
{Analogous results also hold when the possible bid space is discretized
(i.e., players need to bid in integer number of pennies). With a finely enough
discretized bid space, the players could approximately follow the bidding
strategies used in the above proofs, as well as in the proofs in
Appendix \ref{appendix:bayes}. The Nash property then implies that the
same bound holds at the equilibria with a small loss due to the
discretization. Recall that this case is both of practical relevance,
and using a result of Athey \cite{Athey} and Reny \cite{Reny}, in the discrete
case if player types and quality factors are drawn independently, the existence of pure strategy equilibria that are monotone in the types is also
guaranteed.}

{In order to illustrate this point, we {show how to} adapt Lemma
\ref{smoothness-lemma} and Theorem \ref{thm:cbpoa} to the case where} {possible
bids are discrete. Assume bids $b_i$ must be in the finite set $T_{\epsilon, K} =
\{0, \epsilon, 2 \epsilon, \hdots, K \epsilon \}$ for some large integer $K$. We
{also need} to assume that $\epsilon$ is small compared to the
{valuations}. We will assume that $\epsilon < \frac{1}{e}$ and all types
in the support of the
distribution are $v_i \geq 1$. We show that:}

\begin{lemma}\label{smoothness-lemma-discrete}
The GSP auction game with {discretized} bid is
$((1-\frac{1}{e})(1-e\epsilon),1-e\epsilon)$-semi-smooth assuming $\epsilon <
\frac{1}{e}$ and $v_i \geq 1$ for all $i$. That is, there is a deviation $b_i'$
from the discrete space $T_{\epsilon, K}$ that {satisfies} the
semi-smoothness inequality.
\end{lemma}

\begin{proof}
The proof follows from a small modification of Lemma \ref{smoothness-lemma}.
There we considered the deviation where a player with {valuation} $v_i$ samples
a
bid $b'_i$ from the distribution with density $f(y) = \frac{1}{v_i - y}$ for $y
\in [0, (1-\frac{1}{e})v_i]$. In this setting{,
bids} must lie in $T_{\epsilon, K}$, so we use a rounded version instead:
$\hat{b}'_i = \epsilon \cdot \lceil \frac{y}{\epsilon} \rceil$. This change
increases the probability that $\gamma_i \hat{b}'_i \ge \gamma_{\pi(\nu(i))}
b_{\pi(\nu(i))}$, but decreases the expression $v_i-b'_i$ inside the integral.
{This decrease, however, is bounded since it holds that}
$${\min_y} \frac{v_i-\epsilon \cdot \lceil \frac{y}{\epsilon} \rceil}{v_i-y}\ge
(1- \frac{\epsilon e}{v_i})\ge 1- \epsilon e.$$
Using the same calculation as in the proof of Lemma \ref{smoothness-lemma}, we get that
\begin{eqnarray*}
\E_{\hat{b}'_i} [u_i(\hat{b}'_i, \bidsmi, \gamma)] & \geq &  (1-\epsilon e) \left( 1-\frac{1}{e} \right)
\alpha_{\nu(i)}
\gamma_i v_i - (1-\epsilon e) \alpha_{\nu(i)} \gamma_{\pi(\nu(i))}
b_{\pi(\nu(i))}.
\end{eqnarray*}
\end{proof}

{Using this Lemma \ref{smoothness-lemma-discrete} we {immediately} get
the following theorem.}

{
\begin{theorem}
The price of anarchy of the Generalized Second Price auction with uncertainty
and discretized bids is at most $(1+\frac{1}{1-\epsilon 
e})\cdot (1-\frac{1}{e})^{-1}=3.16+O(\epsilon)$, assuming that bids lie on
$T_{\epsilon, K}=\{0, \epsilon, 2 \epsilon, \hdots, K \epsilon \}$, for a large
integer $K$ and $\epsilon <1/e$, and that, for each player $i$, $v_i\geq 1$. 
\end{theorem}
}

\section{Pure Nash Equilibria in the Full Information Setting}\label{sec:fullinfo}

In this section we turn our attention to the full information
setting, where the quality factors $\gamma$ are
fixed and common knowledge. Without loss of generality we can assume that $\gamma_1 v_1 \geq
\gamma_2 v_2 \geq \hdots \geq \gamma_n v_n$. In this setting the strategy of a
player is a single bid $b_i \in [0, v_i]$, again assuming that players do not overbid. Our main result in this setting
is the following:

\begin{theorem} \label{thm:full_info_bound}
The (pure) price of anarchy of the Generalized Second Price auction in the full information setting
is at most $1.282$. In other words, for any fixed click-through-rates $\alpha$,
valuation profile $\vals$, and quality factors $\gamma$, if $\bids$ is a bid
profile in pure Nash equilibrium, then $SW(\pi(\bids), \vals) \geq \frac{1}{1.282} \cdot OPT(\vals)\approx 0.78 \cdot OPT(\vals) $.
\end{theorem}

The bound above is very close to being tight, since we can exhibit an example with $3$
players and $3$ slots for which there is an equilibrium where the gap between
the optimal social welfare and the social welfare in equilibrium is $1.259$. Also, we can show
the following slightly stronger bound for a small number of players and slots.
Notice however that the bound in Theorem \ref{thm:full_info_bound} holds
regardless of the number of slots.

\begin{theorem} \label{few_slots_thm}
For $2$ players and $2$ slots, the price of anarchy is exactly $1.25$. For $3$
players and $3$ slots, the price of anarchy is exactly $1.259$. By
\emph{exactly} we mean that there is a particular GSP auction game with an equilibrium matching
this bound.
\end{theorem}
\begin{proof}
Here we give an example with two slots that yields price of anarchy $1.25$. In Appendix \ref{appendix:fullinfo}, we show that this is worst possible, and show the bound for 3 slots.

For two slots, consider an example with two players with valuations 1 and $1/2$ respectively, quality factors $\gamma_1=\gamma_2=1$, and two slots with $\alpha_1 = 1$ and $\alpha_2 = 1/2$. The bids $b_1=0$ and $b_2=1/2$ are at equilibrium, resulting in a social welfare of $1$, while the optimal social welfare is $1.25$.
\end{proof}

The full proof of Theorem \ref{thm:full_info_bound} can be found in Appendix \ref{appendix:fullinfo}.  Here instead, we present
the proof of a weaker bound that
highlights the intuition underlying our result that GSP equilibria have good
social welfare properties.

\begin{theorem} \label{bound2}
The (pure) price of anarchy of the Generalized Second Price auction in the full information setting is at most $2$.
\end{theorem}

The proof is based on the concept of \emph{weakly feasible allocations}.
Recall that each
bid profile $\bids$ defines an allocation $\pi$ that is a mapping from
slots to players $\pi:[n] \rightarrow [n]$.

\begin{defn}[weakly feasible allocations]
 We say that an allocation $\pi$ is \emph{weakly feasible} if the following holds
for each pair $i,j$ of slots:
\begin{equation} \label{weak_equilibrium}
 \frac{\alpha_j}{\alpha_i} + \frac{\gamma_{\pi(i)} v_{\pi(i)}
}{\gamma_{\pi(j)} v_{\pi(j)}} \geq 1.
\end{equation}
\end{defn}
We also use the term {\em weak feasibility condition} to refer to inequality (\ref{weak_equilibrium}).

{The concept of weakly feasible allocations is a relaxation of the
concept of Nash equilibrium. We adopt this terminology to denote it is a
weakening of the feasibility conditions for Nash equilibrium. This concept}
{encapsulates the fact that an
allocation in equilibrium cannot be too far from the
optimal. The optimal allocation is such that $\pi(i) = i$, since both
$\{\alpha_i\}$ and $\{\gamma_i v_i \}$ are sorted. If an allocation is not
optimal, then two slots $i < j$ have advertisers assigned to them such that
$\pi(i) > \pi(j)$, i.e., they are assigned in the wrong order. Equation
(\ref{weak_equilibrium}) implies that at least one of the two ratios is at
least $1/2$, and hence whenever advertisers are assigned in the non-optimal
order, then either
(i) the two
advertisers have similar effective values for a click, or (ii)
the click-through-rates of the two slots are not very
different; in either case their relative
order does not affect the social welfare very much.}

\begin{lemma} \label{theorem_weak_formulation}
If $\bids$ is a Nash equilibrium of the GSP auction game, then the induced allocation
$\pi$ satisfies the weak feasibility condition.
\end{lemma}

\begin{proof}
If $j \le i$ the inequality is obviously true. Otherwise consider the player $\pi(j)$ in slot $j$.
Since $\bids$ is a Nash equilibrium, the player in slot $j$ is happy with her outcome
and does not want to increase her bid to take slot $i$, so:
$\alpha_j (\gamma_{\pi(j)} v_{\pi(j)} - \gamma_{\pi(j+1)} b_{\pi(j+1)}) \geq
\alpha_i (\gamma_{\pi(j)} v_{\pi(j)} - \gamma_{\pi(i)} b_{\pi(i)})$
since $b_{\pi(j+1)} \geq 0$ and $b_{\pi(i)} \leq v_{\pi(i)}$ then:
$\alpha_j \gamma_{\pi(j)} v_{\pi(j)} \geq \alpha_i (\gamma_{\pi(j)} v_{\pi(j)} -
\gamma_{\pi(i)} v_{\pi(i)})$.
\end{proof}

Given Lemma \ref{theorem_weak_formulation}, the proof of Theorem \ref{bound2}
follows almost directly:

\begin{proofof}{Theorem \ref{bound2}}
 Taking $j = \sigma(i)$ in the definition of weakly feasible allocations, we get
that: $\alpha_{\sigma(i)} \gamma_{i} v_{i}  + \alpha_i \gamma_{\pi(i)}
v_{\pi(i)} \geq \alpha_i \gamma_{i} v_{i} $. Now, summing this for each player
$i$, we get
$$2 \cdot SW(\pi(\bids), \vals) = \sum_i \alpha_{\sigma(i)} \gamma_{i} v_{i}  + \sum_i
\alpha_i \gamma_{\pi(i)} v_{\pi(i)} \geq \sum_i{\alpha_i \gamma_{i} v_{i}} =
OPT(\vals).$$
\end{proofof}

To prove Theorem \ref{thm:full_info_bound} we proceed by induction on the number
of slots. Given an allocation $\pi$, consider the directed graph $G(\pi)$ that
has one node for each slot, and a directed edge for each advertiser $i$ that
connects the node corresponding to slot $i$ to the node corresponding to slot
$\pi^{-1}(i)$. When the allocation is optimal, this graph consists of
self-loops. In general, $G(\pi)$ consists of a set of disjoint cycles, {however,
without loss of generality,} we can assume $G(\pi)$ is a single cycle. We obtain
the improved bound by considering four nodes in the neighborhood of node $1$ in
this cycle, and separately considering cases depending on the order of the
{effective values} of the corresponding players. The details of the proof
can be found in Appendix \ref{appendix:fullinfo}.

\section{Quality of Learning Outcomes in GSP}
\label{sec.learning}

In this section, we bound the average quality of outcomes in a repeated
play of a GSP auction game where players employ strategies that guarantee no external regret.
In both the full information setting and the setting with uncertainty, we can reduce the
problem over declaration sequences to a problem over distributions.
This will allow us to adapt our earlier bounds on the price of
anarchy from Sections \ref{sec:uncertainty} and \ref{sec:fullinfo}
to bound the price of total anarchy.
%
As in previous sections, we show simple and intuitive bounds in this section, and defer
improved and more complex bounds to the appendix.

\newcommand{\ddist}{\mathbf{{D}}}

\subsection{Learning in the full information setting}
\label{sec.learning-full}

We will first focus upon the full information setting of the GSP auction.
Recall that, in this model, the valuation profile $\vals$ and quality factors
$\gamma$ are fixed and common knowledge. As in the previous section, we will assume
that $\gamma_1 v_1 \geq \gamma_2 v_2 \geq \hdots \geq \gamma_n v_n$.

We will begin by proving a relationship between the price of total anarchy and the
set of \emph{coarse correlated equilibria} for the GSP auction in the full information
model.
Given a valuation profile $\vals$, a distribution $\ddist$ over bid profiles is called a
coarse correlated equilibrium if
$$\E_{\bids \sim \ddist} [u_i(\bids)] \geq \E_{\bids \sim \ddist}
[u_i(b'_i, \bids_{-i})], \forall i, b'_i.$$
As we shall show, the price of total anarchy can be bounded by considering the social welfare
generated at any coarse correlated equilibrium.

\begin{lemma}\label{lem:total-anarchy-fullinfo}
The price of total anarchy in the full information setting is at most
$$\sup_{\vals, \ddist \in CCE} \frac{OPT(\vals)}{\E_{\bids \sim
\ddist}[SW(\pi(\bids), \vals)]}$$ where $CCE$ is the set of coarse correlated equilibria.
\end{lemma}
\begin{proof}
Consider a declaration sequence $D = (\bids^1, \hdots, \bids^t,
\hdots)$ in the full information case. For each $T$ let $\ddist^T$ be the distribution over bid
profiles where each $\bids^t$ for $t \leq T$ is drawn with probability
$\frac{1}{T}$. Proving that the price of total anarchy is bounded by $\eta$ is
equivalent to showing that:
$$\liminf_T \E_{\bids \sim \ddist^T} [ SW(\pi(\bids), \vals) ]
\geq \frac{1}{\eta} OPT(\vals).$$
Since the set of all possible bid profiles is compact, one needs to
prove that for all distributions $\ddist$ such that there is a subsequence
of $\{\ddist^T\}_T$ converging in distribution to $\ddist$ we have:
$$\E_{\bids \sim \ddist} [ SW(\pi(\bids), \vals) ]
\geq \frac{1}{\eta} OPT(\vals).$$
It is therefore sufficient to show that such a $\ddist$ is a coarse correlated
equilibrium.  We note that the fact that the declaration sequence $D$
minimizes external
regret implies that, for each distribution $\ddist$ which can be written as
the limit of a subsequence of $\{\ddist^T\}_T$, it holds that:
$$\E_{\bids \sim \ddist} [u_i(\bids)] \geq \E_{\bids \sim \ddist}
[u_i(b'_i, \bids_{-i})], \forall i, b'_i$$
as required.
\end{proof}

Using this connection to coarse correlated equilibria, we are able to
obtain a bound of $2.310$ on the price of total anarchy of the GSP auction.  

\begin{theorem}
\label{thm.cce}
The price of total anarchy of the Generalized Second Price auction in the full information setting is at most $2.310$.
\end{theorem}

A full proof of Theorem \ref{thm.cce} appears in
Appendix \ref{appendix:learning}.  We now present a simpler proof of the following
weaker bound, which captures some of the intuition behind the proof of Theorem \ref{thm.cce}.

\begin{theorem} \label{thm:full-info-learning-3}
The price of total anarchy of the Generalized Second Price auction in the full information setting is at most
$3$.
\end{theorem}
\begin{proof}
The proof can be thought of as an improved version of the price of anarchy bound
based on the fact that the GSP auction game is $(1/2,1)$-semi-smooth. We consider a distribution
$\ddist$ which corresponds to a coarse correlated equilibrium. All expectations
in the following are taken with respect to $\bids \sim \ddist$. Recall the
outline of the bound of $4$ on the price of anarchy based on the fact that the GSP auction game is
$(1/2,1)$-semi-smooth.  We considered a possible deviation for player $i$ with
valuation $v_i$  to bid $b'_i=v_i/2$, and concluded the bound $u_i(b'_i, \bidsmi)\ge
\alpha_i \gamma_i v_i/2 -\alpha_i \gamma_{\pi(i)} b_{\pi(i)}$ in the proof of
Lemma~\ref{smoothness-lemma}. We use the no-regret inequality directly, to get
that
$$\E[u_i(\bids)]\ge \alpha_i \gamma_i v_i/2 -\E[\alpha_i \gamma_{\pi(i)} b_{\pi(i)}].$$
Using that $b_{\pi(i)} \leq v_{\pi(i)}$, and summing over all players we get a
bound of 4 on the price of total anarchy as was done in Lemma
\ref{lem.smoothness-to-bpoa}.

Here we improve this bound by adding two new ideas. First, note that for all
slots except the top one $\alpha_i \gamma_{\pi(i)} b_{\pi(i)}$ is a lower
bound to the payment of the player in slot $i-1$. The social welfare is the sum of
player utilities and the payments. The inequality states that in expectation the
utility of player $i$ plus the payment of the player in slot $i-1$ is at least
$\alpha_i \gamma_i v_i/2$, i.e., half of the social welfare contributed by player $i$
in the efficient solution. To turn this into a bound on social welfare, we need
to handle player $1$ differently, as $\alpha_1 \gamma_{\pi(1)} b_{\pi(1)}$
does not correspond to any payment.

The second observation is that for player $1$ we can obtain a stronger bound on
her utility by considering the deviation $b'_1=v_1$. For other players such a
high bid would endanger them to get a slot much higher than their slot in the
optimum at a very high price. But player 1 already gets the best slot in the
efficient solution. Deviating to $b_1'=v_1$ will give the player the top slot,
and hence utility $\alpha_1\gamma_1 v_1 -\alpha_1 \gamma_{\pi(1)} b_{\pi(1)}$. Now using
the no-regret property we get
$$\E[u_1(\bids)]\ge \alpha_1 \gamma_1 v_1 -\E[\alpha_1 \gamma_{\pi(1)}
b_{\pi(1)}].$$

By summing over all players and writing the social welfare as the sum of utilities plus the total payments, we get:
$$\begin{aligned}
& \E[SW(\pi(\bids), \vals)]  = \E [\sum_i{u_i(\bids)}] +
\E[\sum_i{\alpha_i\gamma_{\pi(i+1)} b_{\pi(i+1)}}] \\
& \qquad \geq \frac{1}{2}\E[u_1(\bids)] + \sum_{i\geq 2} \E [{u_i(\bids)}] +
\E[\sum_{i\geq2}{\alpha_i \gamma_{\pi(i)} b_{\pi(i)}}] \\
& \qquad \geq  \frac{\alpha_1 \gamma_1 v_1}{2} -
\frac{\E[\alpha_1\gamma_{\pi(1)} b_{\pi(1)}]}{2} + \sum_{i\geq2} \frac{\alpha_i
\gamma_i v_i}{2} - \sum_{i\geq2}{\E[\alpha_i \gamma_{\pi(i)}b_{\pi(i)}]}
+ \sum_{i\geq 2}{\E[\alpha_i \gamma_{\pi(i)}b_{\pi(i)}]}\\
& \qquad = \frac{1}{2} OPT(\vals) - \frac{1}{2}\E[\alpha_1 \gamma_{\pi(1)}
b_{\pi(1)}].
\end{aligned}$$
Since $\E[SW(\pi(\bids), \vals)]\geq \E[\alpha_1 \gamma_{\pi(1)}
b_{\pi(1)}]$, we obtain that $\E[SW(\pi(\bids), \vals)] \geq
\frac{1}{3} OPT(\vals)$.
\end{proof}

\subsection{Learning with uncertainty}\label{subsec:learning_with_uncertainty}

Let us now turn to the model of learning outcomes with uncertainty.
As in the full information model, we can define a Bayesian version of
the coarse correlated equilibrium.  A \emph{Bayesian coarse correlated
equilibrium} is a
joint distribution $(\vals, \gamma, \bids)$ whose $(\vals,
\gamma)$-marginals are $(\mathbf{F}, \mathbf{G})$ and satisfies the following
property:
$$\E_{(\vals,\gamma,\bids)} [u_i(\bids,\gamma) \vert v_i] \geq
\E_{(\vals,\gamma,\bids)} [u_i(b'_i, \bids_{-i}, \gamma) \vert
v_i], \forall i, v_i, b'_i. $$
Similarly to Lemma \ref{lem:total-anarchy-fullinfo}, we can show that the
price
of total anarchy with uncertainty can be bounded by considering the social welfare
generated at any Bayesian coarse correlated equilibrium.

\begin{lemma}
\label{lem:total-anarchy-uncertainty}
{Assuming that the distribution over types has finite support,} the
price of total anarchy with uncertainty is at most
$$\sup_{\mathbf{F},\mathbf{G}, \ddist(\cdot) \in CCE}
\frac{\E_{\vals,\gamma} OPT(\vals,\gamma)}{\E_{\vals,\gamma,\bids \sim
\ddist(\vals)}[SW(\pi(\bids), \vals, \gamma)]}$$ where $CCE$ is the set of
Bayesian coarse correlated equilibria.
\end{lemma}



\proofsk{
{The proof follows the same lines as the proof of Lemma
\ref{lem:total-anarchy-fullinfo}.
{For each $t \geq 1$, $(\vals^t, \gamma^t, \bids^t)$ is the tuple of profiles
corresponding to round $t$.}
Since the distribution over types has finite
support, there is almost surely some $T_0$ such that, for each type profile
$\tilde{\vals}$ in the support of $\mathbf{F}$, there is $t \leq T_0$ such that $\vals^t =
\tilde{\vals}$. For each $T \geq T_0$, {let $\ddist^T$ be the joint distribution
on $(\vals,
\gamma, \bids)$ that samples $t$ uniformly} from $\{1, 2, \hdots, T\}$ and
outputs $(\vals^t, \gamma^t, \bids^t)$. This defines a sequence of distributions
{$\{\ddist^T\}_{T \geq T_0}$}. Now, it is enough to observe that each convergent
subsequence
converges to a Bayesian coarse correlated equilibrium. Therefore the price of
total anarchy is bounded by the price of anarchy over Bayesian coarse correlated
equilibria.}}

\begin{remark}
 Since our theorems hold in the limit as $T$ goes to infinity, they
do not depend on the speed of learning -- which we can define as the speed in
which subsequences of $\{\ddist^T\}_T$ converge to a Bayesian coarse correlated
equilibrium in the proof above.  {The rate of convergence depends on
the specific learning methods being used by the players.  However, the reader might notice
that, regardless of the learning methods used, the speed of learning will depend on the time
required for the empirical distribution of $\vals, \gamma$ to resemble the real distribution.
The speed in which this
happens is controlled, for example, by the Central Limit Theorem.  Also, if players observe
only realized payoffs each round (rather than expectations), one would expect low click-through
rates to increase the amount of time needed for learning, since more rounds will be required to
accurately estimate expected outcomes.  See Auer et al. \cite{AU95} for a more detailed
discussion on the speed of convergence of no-regret algorithms with limited feedback.}
\end{remark}

The arguments in the proof of Lemma \ref{lem.smoothness-to-bpoa} can be used
with essentially no change to show that $(\lambda,\mu)$-semi-smoothness implies
a bound of $(\mu+1) / \lambda$ to the price of total anarchy with uncertainty.
From this, we know that:

\begin{theorem}
\label{thm.bcce}
 The price of total anarchy of the Generalized Second Price auction with uncertainty is bounded by $3.164$.
\end{theorem}

In Appendix \ref{appendix:bayes} we present an improved result for the Bayesian
price of anarchy that also extends to the following improved
bound for learning outcomes.
\begin{theorem}
\label{thm.bcce-improved}
 The price of total anarchy of the Generalized Second Price auction with uncertainty is bounded by $2.927$.
\end{theorem}

\bibliographystyle{abbrv}


\appendix
\comment{
\section{GSP auction games are not smooth}\label{appendix:smoothness}

\begin{theorem} The GSP auction game is not $(\lambda, \mu)$-smooth for any
parameters $\lambda, \mu$. \end{theorem}

\begin{proof}
Consider a GSP auction game with $2$ slots with click-through-rates $1$ and $\alpha < 1$, and
two advertisers with valuations $1$ and $v$. Let $s = (b_1, b_2)$ and $s^* = (b_3,
b_4)$ where $1 > v > b_2 > b_3 > b_4 > b_1 \geq 0$. For this case, the expression
$\sum_i u_i (s_i^*, s_{-i}) \geq \lambda SW(s^*) - \mu SW (s)$ (according to Roughgarden's definition of smoothness \cite{roughgarden}) becomes:
$$ \alpha(1-0) + 1(v - b_1) \geq \lambda( 1 + \alpha v) - \mu (v +
\alpha),$$
which yields
$$(1 + \mu) (\alpha + v) \geq \lambda (1 + \alpha v),$$
for $b_1\geq 0$. Given any $\lambda$, $\mu > 0$, there exist arbitrarily small $\alpha$ and $v$ so that the above inequality is violated; hence, the GSP auction game is not $(\lambda, \mu)$-smooth.
\end{proof}
}
\section{Improved Bounds for Games with Uncertainty}\label{appendix:bayes}

\newcommand{\iannis}[1]{\noindent\textbf{iannis:}\marginpar{****}%
\textit{{#1}}\textbf{:iannis}}%

In this section, we prove Theorems \ref{thm:cbpoa_main} and \ref{thm.bcce-improved}. The idea of the proof is analogous to our proof of the bound of $3.164$ in Section \ref{sec:uncertainty}, based on semi-smoothness, but will use a modification of semi-smoothness specially tailored to the GSP auction game, analogous to the way we modified the simple bound of 4 derived using the $(1/2,1)$-semi-smoothness of GSP to a bound of 3 on the price of total anarchy for the full information case in Section \ref{sec.learning}. We handle the case when a player has the highest effective value separately, and show that there exists a bidding profile $\bids'$ such that the following inequality holds.

\begin{equation}\label{semi-smooth-weak}
\E[\sum_i u_i(b'_i(v_i), \bidsmi,\gamma)] \geq \beta\E[OPT(\mathbf{v},\gamma)]-(1+\delta)\sum_i \E[\alpha_i\gamma_{\pi(i)}b_{\pi(i)}]+\E[\alpha_1\gamma_{\pi(1)}b_{\pi(1)}].
\end{equation}
This inequality is analogous but weaker than claiming that GSP is $(\beta,\delta)$-semi-smooth, yet we will show that it implies that the price of anarchy (and the price of total anarchy) is bounded by $\frac{1+\delta}{\beta}$. This connection is stated in the next lemma.

\begin{lemma}\label{lem:ineq-implies-poa-bound}
Assume that for every GSP auction game there is a bidding profile $\bids'$ and parameters $\beta, \delta>0$ such that inequality (\ref{semi-smooth-weak}) holds for any strategy profile $\bids$. Then, the price of anarchy of the Generalized Second Price auction with uncertainty is at most $\frac{1+\delta}{\beta}$. The same bound applies to the price of total anarchy with uncertainty as well.
\end{lemma}

\begin{proof}
Consider a Nash equilibrium bidding profile $\bids$. Clearly, $\E[u_i(\bids,\gamma)] \geq \E[u_i(b'_i(v_i),\bidsmi,\gamma)]$ by selecting the bidding profile $\bids'$ as in inequality (\ref{semi-smooth-weak}). We use this inequality and the fact that the social welfare is the sum of the expected utilities of the advertisers plus the total payments to get
\begin{eqnarray*}
\E[SW(\pi(\mathbf{b(v)},\gamma),\mathbf{v},\gamma)] &=&\mathbb{E}[\sum_i{u_i(\mathbf{b},\gamma)}] +\mathbb{E}[\sum_i{\alpha_i\gamma_{\pi(i+1)}b_{\pi(i+1)}}]\\
&\geq &\E[\sum_i u_i(b'_i(v_i), \bidsmi,\gamma)] + \mathbb{E}[\sum_i{\alpha_i\gamma_{\pi(i+1)}b_{\pi(i+1)}}]\\
&\geq &\beta\E[OPT(\mathbf{v},\gamma)]-(1+\delta)\sum_i \E[\alpha_i\gamma_{\pi(i)}b_{\pi(i)}]+\E[\alpha_1\gamma_{\pi(1)}b_{\pi(1)}]\\
&& {}+{}\sum_{i\geq 2}{\E[\alpha_i\gamma_{\pi(i)}b_{\pi(i)}]}\\
&=& \beta\E[OPT(\mathbf{v},\gamma)]-\delta\sum_i \E[\alpha_i\gamma_{\pi(i)}b_{\pi(i)}]\\
&\geq& \beta\E[OPT(\mathbf{v},\gamma)]-\delta\E[SW(\pi(\mathbf{b(v)},\gamma),\mathbf{v},\gamma)],
\end{eqnarray*}
which implies that the price of anarchy is at most $\frac{1+\delta}{\beta}$, as desired.
To get the same bound for the price of total anarchy, consider a coarse correlated equilibrium $\bids$ instead of a Nash equilibrium.
\end{proof}

The next lemma (Lemma \ref{lem:f-poa}) connects inequality (\ref{semi-smooth-weak}) to the existence of functions with particular properties which we call $(\beta, \delta)$-bounded functions.

\begin{defn}\label{def:bounded}
Let $\beta,\delta >0$ and  $g:[0,1]\rightarrow \mathbb{R}_+$. Function $g$ is $(\beta, \delta)$-bounded if the following three properties hold:
\begin{eqnarray*}
i)&& \int_0^1 g(y)\ud y \leq 1,\\
ii)&& (1-z) \int_z^1 g(y)\ud y \geq \beta -\delta z, \quad \forall z\in[0,1],\\
iii)&&  \int_z^1 (1-y)g(y)\ud y \geq \beta -(1+\delta) z, \quad \forall z\in[0,1].
\end{eqnarray*}
\end{defn}
Recall that the proof of Lemma \ref{smoothness-lemma} that GSP is $(1-\frac{1}{e},1)$-semi-smooth relied on a random distribution with density $f(y)=\frac{1}{1-y}$ for $y \in [0,(1-\frac{1}{e})]$ and $f(y) = 0$ otherwise, and considered the bid distribution $b'_i=yv_i$ for player $i$ with valuation $v_i$. The improved proof in Lemma \ref{lem:f-poa} uses a $(\beta,\delta)$-bounded function $g$ in place of this $f$.

\begin{lemma}\label{lem:f-poa}
Let $\beta, \delta>0$ be such that a $(\beta, \delta)$-bounded function exists. Then, there is a bidding profile $\bids'$ such that inequality (\ref{semi-smooth-weak}) holds for any strategy profile $\bids$.
\end{lemma}

\begin{proof}
In the proof we consider a GSP auction game with $n$ slots with click-through-rates $\alpha_1\geq \alpha_2 \geq \ldots \geq \alpha_n\geq 0$ and $n$ conservative players with random valuations $v_1, v_2, \ldots, v_n\geq 0$ and random quality factors $\gamma_1, \gamma_2, \ldots, \gamma_n\geq 1$. Let $\mathbf{b}$ denote any bid profile. Also, we denote by $\bids'$ the bid profile such that $b'_i(x)$ is the most profitable deviation for player $i$ when her valuation is $v_i=x$. We will prove inequality (\ref{semi-smooth-weak}) using this definition for $\bids'$.

The proof is long and technical. Before presenting it, we give a high-level overview. We apply the following three steps:
\begin{itemize}
\item Step 1: We focus on advertiser $i$ with valuation $v_i=x$ and obtain a lower bound on her expected utility $\mathbb{E}[u_i(b'_i,\bidsmi,\gamma)|v_i=x]$ when deviating to $b'_i(x)$. The main idea we use here is that the deviation to bid $b'_i(x)$ is more profitable for advertiser $i$ than deviating to the bid $yx$, for every $y\in[0,1]$. This yields infinitely many lower bounds on $\mathbb{E}[u_i(b'_i,\bidsmi,\gamma)|v_i=x]$; we combine them in a single lower bound by taking their weighted average, with weights indicated by the values of a $(\beta,\delta)$-bounded function $g$.
\item Step 2: We further refine the lower bound on $\mathbb{E}[u_i(b'_i,\bidsmi,\gamma)|v_i=x]$. Here, we reason about the slots advertiser $i$ would occupy by deviating to bid $yx$ and the utility she would then have, and we use the properties of $(\beta,\delta)$-bounded functions. We consider slot $1$ and slots $i \ge 2$ separately, as we did in the proof of Theorem \ref{thm:full-info-learning-3}.
\item Step 3: We use the bound obtained in Step 2 in order to compute a lower bound for the total expected utility of all players when deviating to $\bids'$. We first lower-bound the unconditional expected utility of advertiser $i$ and, then, we simply sum the obtained inequalities over all advertisers in order to obtain inequality (\ref{semi-smooth-weak}).
\end{itemize}

\paragraph{Step 1:} Focus on player $i$ and let $x$ be a possible valuation for this player. Let $\beta,\delta>0$ and consider a $(\beta,\delta)$-bounded function $g:[0,1]\rightarrow \mathbb{R}_+$. Using the first property in Definition \ref{def:bounded} for $g$ and the fact that $b'_i(x)$ is the most profitable deviation for advertiser $i$, we have
\begin{eqnarray*}
\E[u_i(b'_i(x),\bidsmi,\gamma)|v_i=x] &\geq & \int_0^1{g(y)\E[u_i(b'_i(x),\bidsmi,\gamma)|v_i=x]\ud y}\\
&\geq &\int_0^1{g(y)\E[u_i(yx,\bidsmi,\gamma)|v_i=x]\ud y}.
\end{eqnarray*}
Given any slot $j$, let $A^{ij}_x$ denote the event that $v_i=x$ and $\nu(i)=j$ and $B^{ij}_x$ denote the event that $\nu(i)=j$ given that $v_i=x$. Using these definitions, we can rewrite the quantity $\E[u_i(yx,\mathbf{b}_{-i},\gamma)|v_i=x]$ for every $y\in [0,1]$ as
\begin{eqnarray*}
\E[u_i(yx,\mathbf{b}_{-i},\gamma)|v_i=x] &=& \sum_{j=1}^n{\E[u_i(yx,\mathbf{b}_{-i},\gamma)|A^{ij}_x]\cdot \prob[B^{ij}_x]}.
\end{eqnarray*}
By the last two (in)equalities, we obtain that
\begin{eqnarray}\nonumber
\E[u_i(b'_i(x),\bidsmi,\gamma)|v_i=x]  &\geq & \int_0^1{g(y)\sum_{j=1}^n{\E[u_i(yx,\mathbf{b}_{-i},\gamma)|A^{ij}_x]\cdot \prob[B^{ij}_x]}\ud y}\\\label{eq:bn-improved-main}
&=& \sum_{j=1}^n{\int_0^1{g(y)\E[u_i(yx,\mathbf{b}_{-i},\gamma)|A^{ij}_x]\ud y}\cdot \prob[B^{ij}_x]}.
\end{eqnarray}

\paragraph{Step 2:} Our purpose now is to refine the lower bound provided by inequality (\ref{eq:bn-improved-main}). Let $\pi^i(\bidsmi,i)$ be the player with the $i$-th highest effective bid in $\bidsmi$.

First consider slot 1 separately. Assume that the event $A^{i1}_x$ is true, i.e., $v_i=x$ and $\nu(i)=1$. We will first lower-bound the quantity $\E[u_i(yx,\mathbf{b}_{-i},\gamma)|A^{i1}_x]$ for every $y\in [0,1]$. By deviating to bid $yx$, player $i$ is allocated the first slot whenever $\gamma_i yx>\gamma_{\pi^i(1)}b_{\pi^i(1)}$; in this case, player $i$ has utility at least $\alpha_1 (\gamma_i x - \gamma_{\pi^i(1)} b_{\pi^i(1)})$. Hence,
\begin{eqnarray*}
\E[u_i(yx,\mathbf{b}_{-i},\gamma)|A^{i1}_x] &\geq & \E[\alpha_1(\gamma_i x-\gamma_{\pi^i(1)}b_{\pi^i(1)})\one{\gamma_i y x > \gamma_{\pi^i(1)}b_{\pi^i(1)}}|A^{i1}_x].
\end{eqnarray*}
We set $z=\frac{\gamma_{\pi^i(1)}b_{\pi^i(1)}}{\gamma_i x}$ and use this last inequality to obtain
\begin{eqnarray} \nonumber
\int_0^1{g(y)\E[u_i(yx,\mathbf{b}_{-i},\gamma)|A^{i1}_x]\ud y} &\geq & \int_0^1{g(y)\cdot \E[\alpha_1\gamma_i x(1-z)\one{y > z}|A^{i1}_x]\ud y}\\\nonumber
&=& \E[\alpha_1\gamma_i x (1-z)\int_0^1{g(y)\one{y > z}\ud y}|A^{i1}_x]\\\nonumber
&=& \E[\alpha_1 \gamma_i x (1-z)\int_z^1{g(y)\ud y}|A^{i1}_x]\\
\label{eq:bn-improved-slot1}
&\geq& \E[\alpha_1\left(\beta\gamma_i x - \delta\gamma_{\pi^i(1)} b_{\pi^i(1)}\right)|A^{i1}_x]
\end{eqnarray}
where the second inequality follows by the second property of Definition \ref{def:bounded} for function $g$ (and using the definition of $z$).

Now, assume that the event $A^{ij}_x$ is true for $j\geq 2$, i.e., $v_i=x$ and $\nu(i)=j$. We will lower-bound the quantity $\E[u_i(yx,\mathbf{b}_{-i},\gamma)|A^{ij}_x]$ for every $y\in [0,1]$. By deviating to bid $yx$, player $i$ is allocated slot $j$ (or a higher one) whenever $\gamma_i yx>\gamma_{\pi^i(j)}b_{\pi^i(j)}$; in this case, player $i$ has utility at least $\alpha_j\gamma_i x(1-y)$. Hence,
\begin{eqnarray*}
\mathbb{E}[u_i(yx,\mathbf{b}_{-i},\gamma)|A^{ij}_x] &\geq & \E[\alpha_{j} \gamma_i x(1-y)\one{\gamma_i yx > \gamma_{\pi^i(j)}b_{\pi^i(j)}}|A^{ij}_x].
\end{eqnarray*}
We set $z=\frac{\gamma_{\pi^i(j)}b_{\pi^i(j)}}{\gamma_i x}$ and use this last inequality to obtain
\begin{eqnarray}\nonumber
\int_0^1{g(y)\E[u_i(yx,\mathbf{b}_{-i},\gamma)|A^{ij}_x]\ud y}&\geq & \int_0^1{g(y)\E[\alpha_j\gamma_i x(1-y)\one{y>z}|A^{ij}_x]\ud y}\\\nonumber
&=& \E[\alpha_j\gamma_i x\int_z^1{(1-y)g(y)\ud y}|A^{ij}_x]\\\label{eq:bn-improved-slotj}
&\geq& \E[\alpha_j\left(\beta\gamma_i x - (1+\delta)\gamma_{\pi^i(j)} b_{\pi^i(j)}\right)|A^{ij}_x].
\end{eqnarray}
The second inequality follows by the third property of Definition \ref{def:bounded} for function $g$ (and using the definition of $z$).

We now use inequality (\ref{eq:bn-improved-main}) together with the lower bounds for $\int_0^1{g(y)\E[u_i(yx,\mathbf{b}_{-i},\gamma)|A^{ij}_x]\ud y}$ obtained in (\ref{eq:bn-improved-slot1}) and (\ref{eq:bn-improved-slotj}). We have
\begin{eqnarray*}
\E[u_i(b'_i(x),\bidsmi,\gamma)|v_i=x] &\geq & \E[\alpha_1\left(\beta\gamma_i x - \delta\gamma_{\pi^i(1)} b_{\pi^i(1)}\right)|A^{i1}_x]\cdot \prob[B^{i1}_x]\\
&&{}+{} \sum_{j=2}^n{\E[\alpha_j\left(\beta\gamma_i x - (1+\delta)\gamma_{\pi^i(j)} b_{\pi^i(j)}\right)|A^{ij}_x]\cdot \prob[B^{ij}_x]}\\
&=& \beta\sum_{j=1}^n{\E[\alpha_j \gamma_i x|A^{ij}_x]\cdot \prob[B^{ij}_x]}-\delta \E[\alpha_1 \gamma_{\pi^i(1)} b_{\pi^i(1)}|A^{i1}_x]\cdot \prob[B^{i1}_x]\\
&& {}-{}(1+\delta)\sum_{j=2}^{n}{\E[\alpha_j\gamma_{\pi^i(j)} b_{\pi^i(j)}|A^{ij}_x]\cdot \prob[B^{ij}_x]}.
\end{eqnarray*}

\paragraph{Step 3:} We can now bound the unconditional expected utility of player $i$ when deviating to strategy $b'_i(v_i)$ by integrating over the range of valuations for player $i$ and using the last inequality obtained in Step 2. In the following we use $f_{v_i}(x)$ to denote the probability density function  of the random variable $v_i$. We have
\begin{eqnarray*}
\mathbb{E}[u_i(b'_i(v_i),\bidsmi,\gamma)]&=&\int_0^\infty {\mathbb{E}[u_i(b'_i(v_i),\bidsmi,\gamma)]\cdot f_{v_i}(x)}\ud x\\
&\geq&\beta \sum_{j=1}^{n}  \int_0^\infty {\E [\alpha_{j} \gamma_iv_i|A^{ij}_x]\cdot \prob[B^{ij}_x]\cdot f_{v_i}(x)}\ud x\\
&&{}-{}\delta \int_0^\infty { \E[\alpha_1\gamma_{\pi^i(1)}b_{\pi^i(1)}|A^{i1}_x]\cdot \prob[B^{i1}_x]\cdot f_{v_i}(x)}\ud x\\
&&{}-{}(1+\delta)\sum_{j=2}^{n}\int_0^\infty { \E[\alpha_{j}\gamma_{\pi^i(j)}b_{\pi^i(j)}|A^{ij}_x]\cdot \prob[B^{ij}_x]\cdot f_{v_i}(x)}\ud x.
\end{eqnarray*}
Now, we use the property
\begin{eqnarray*}
\int_0^{\infty}{\mathbb{E}[Z|A^{ij}_x]\cdot \prob[B^{ij}_x]\cdot f_{v_i}(x)\ud x}&=& \mathbb{E}[Z|\nu(i)=j]\cdot\prob[\nu(i)=j],
\end{eqnarray*}
for any random variable $Z$ as well as the fact that $\gamma_{\pi^i(j)}b_{\pi^i(j)}\leq \gamma_{\pi(j)}b_{\pi(j)}$ to obtain that
\begin{eqnarray*}
&& \mathbb{E}[u_i(b'_i(v_i),\bidsmi,\gamma)]\\
& \geq &\beta \sum_{j=1}^{n} \E [\alpha_{j} \gamma_iv_i|\nu(i)=j]\cdot \prob[\nu(i)=j] -\delta \E[\alpha_1\gamma_{\pi^i(1)}b_{\pi^i(1)}|\nu(i)=1]\cdot \prob[\nu(i)=1]\\
&&{}-{}(1+\delta)\sum_{j=2}^{n}\E[\alpha_{j}\gamma_{\pi^i(j)}b_{\pi^i(j)}|\nu(i)=j]\cdot \prob[\nu(i)=j]\\
&\geq & \beta \sum_{j=1}^{n} \E [\alpha_{j} \gamma_iv_i|\nu(i)=j]\cdot \prob[\nu(i)=j]-\delta \E[\alpha_1\gamma_{\pi(1)}b_{\pi(1)}|\nu(i)=1]\cdot \prob[\nu(i)=1]\\
&&{}-{}(1+\delta)\sum_{j=2}^{n}\E[\alpha_{j}\gamma_{\pi(j)}b_{\pi(j)}|\nu(i)=j]\cdot \prob[\nu(i)=j]\\
&=& \beta \sum_{j=1}^{n} \E [\alpha_{j} \gamma_iv_i|\nu(i)=j]\cdot \prob[\nu(i)=j]-(1+\delta)\sum_{j=1}^{n}\E[\alpha_{j}\gamma_{\pi(j)}b_{\pi(j)}|\nu(i)=j]\cdot \prob[\nu(i)=j]\\
&&{}+{}\E[\alpha_1\gamma_{\pi(1)}b_{\pi(1)}|\nu(i)=1]\cdot \prob[\nu(i)=1]\\
&=&\beta \E [\alpha_{\nu(i)} \gamma_iv_i]-(1+\delta)\E[\alpha_{\nu(i)}\gamma_{\pi(\nu(i))}b_{\pi(\nu(i))}]+\E[\alpha_1\gamma_{\pi(1)}b_{\pi(1)}|\nu(i)=1]\cdot \prob[\nu(i)=1].
\end{eqnarray*}

By summing over all players, we obtain inequality (\ref{semi-smooth-weak}). In particular,
\begin{eqnarray*}
\sum_i{\mathbb{E}[u_i(b'_i(v_i),\bidsmi,\gamma)]}&\geq&\beta \sum_i \E [\alpha_{\nu(i)} \gamma_iv_i]-(1+\delta)\sum_i \E[\alpha_{\nu(i)}\gamma_{\pi(\nu(i))}b_{\pi(\nu(i))}]\\
&&{}+{}\sum_i \E[\alpha_1\gamma_{\pi(1)}b_{\pi(1)}|\nu(i)=1]\cdot \prob[\nu(i)=1]\\
&=&\beta\E[OPT(\mathbf{v},\gamma)]-(1+\delta)\sum_i \E[\alpha_i\gamma_{\pi(i)}b_{\pi(i)}]+\E[\alpha_1\gamma_{\pi(1)}b_{\pi(1)}].
\end{eqnarray*}
\end{proof}

Therefore, by Lemmas \ref{lem:ineq-implies-poa-bound} and \ref{lem:f-poa}, in order to bound the price of anarchy, it suffices to find a $(\beta,\delta)$-bounded function such that the ratio $\frac{1+\delta}{\beta}$ is as low as possible. This is the purpose of the following lemma.
\begin{lemma}\label{lem:function}
Consider a function $g:[0,1]\rightarrow \mathbb{R}_+$ defined as follows:
\begin{eqnarray*}
g(y) = \left\{
\begin{array}{ll}
\frac{\kappa}{1-y},& y \in[0,\lambda),\\
\frac{(\kappa-1)(1-\mu)}{(1-y)^2},& y\in [\lambda,\mu),\\
0,& y\in[\mu, 1],
\end{array} \right.
\end{eqnarray*}
where $\kappa>1$ and $1> \mu\geq \lambda \geq 0$ such that $\frac{(\kappa-1)(\mu-\lambda)}{1-\lambda}-\kappa \ln(1-\lambda)\leq 1$, and $(\kappa-1)(1-\mu)\ln{\frac{1-\lambda}{1-\mu}}-(\kappa-1)\mu +\kappa\lambda\geq 0$.
Then, $g(y)$ is an $((\kappa -1)\mu, \kappa-1)$-bounded function.
\end{lemma}

\begin{proof}
We begin by computing $\int_0^1 g(y)\ud y$. It holds that
\begin{eqnarray*}
\int_0^1 g(y)\ud y &=& \int_0^\lambda \frac{\kappa}{1-y}\ud y+\int_\lambda^\mu{\frac{(\kappa-1)(1-\mu)}{(1-y)^2}\ud y} =\frac{(\kappa-1)(\mu-\lambda)}{1-\lambda}-\kappa \ln(1-\lambda)
\leq 1,
\end{eqnarray*}
where the inequality holds by the first assumption concerning $\kappa$, $\lambda$ and $\mu$. Hence, $g$ satisfies the first property of Definition \ref{def:bounded}.

For the second property of Definition \ref{def:bounded} it suffices to prove that
\begin{eqnarray*}
(1-z) \int_z^1 g(y)\ud y+(\kappa -1)(z-\mu)&\geq &0, \quad \forall z\in[0,1].
\end{eqnarray*}
We distinguish between three cases depending on $z$. First, we consider the case that $z\in [\mu, 1]$. We have
\begin{eqnarray*}
(1-z) \int_z^1 g(y)\ud y+(\kappa -1)(z-\mu)
&=&(\kappa-1)(z-\mu) \geq 0,
\end{eqnarray*}
where the inequality holds since $z\in [\mu,1]$ and $\kappa>1$. For $z\in (\lambda, \mu)$ we have
\begin{eqnarray*}
(1-z) \int_z^1 g(y)\ud y+(\kappa -1)(z-\mu)&=&(1-z) \int_z^\mu \frac{(\kappa-1)(1-\mu)}{(1-y)^2}\ud y+(\kappa -1)(z-\mu)
= 0.
\end{eqnarray*}
Finally, for $ z\in [0,\lambda]$ we have
\begin{eqnarray*}
&&(1-z) \int_z^1 g(y)\ud y+(\kappa -1)(z-\mu)\\
&=&(1-z)\int_z^\lambda \frac{\kappa}{1-y}\ud y+(1-z) \int_\lambda^\mu \frac{(\kappa-1)(1-\mu)}{(1-y)^2}\ud y+(\kappa -1)(z-\mu)\\
&=&(1-z)\kappa \ln{\frac{1-z}{1-\lambda}}+\frac{(1-z)(\kappa-1)(\mu-\lambda)}{1-\lambda}+(\kappa-1) (z-\mu)\\
&\geq& (\kappa-1)(\mu-\lambda)+(\kappa-1) (\lambda-\mu)\\
&=&0,
\end{eqnarray*}
where the inequality follows by the fact that the derivative with respect to $z$ is negative for $z\in [0,\lambda]$. Hence, it holds that $g$ satisfies the second property of Definition \ref{def:bounded}.

It remains to prove that $g$ satisfies the third property of Definition \ref{def:bounded}. Similarly, it suffices to prove that
\begin{eqnarray*}
\int_z^1 (1-y)g(y)\ud y-(\kappa -1)\mu+\kappa z&\geq &0, \quad \forall z\in[0,1].
\end{eqnarray*}
Again, we distinguish between three cases depending on $z$. First, we consider the case that $z\in [\mu , 1]$. We have
\begin{eqnarray*}
\int_z^1 (1-y)g(y)\ud y-(\kappa -1)\mu+\kappa z
&=&{}-{}(\kappa -1)\mu+\kappa z \geq \mu \geq 0,
\end{eqnarray*}
where the first inequality follows since $z\in [\mu,1]$. For $z\in [\lambda,\mu)$ we have
\begin{eqnarray*}
\int_z^1 (1-y)g(y)\ud y-(\kappa -1)\mu+\kappa z&=&\int_z^\mu \frac{(\kappa-1)(1-\mu)}{(1-y)^2}\ud y-(\kappa -1)\mu+\kappa z\\
&=&(\kappa-1)(1-\mu)\ln{\frac{1-z}{1-\mu}}-(\kappa -1)\mu+\kappa z\\
&\geq&(\kappa-1)(1-\mu)\ln{\frac{1-\lambda}{1-\mu}}-(\kappa -1)\mu+\kappa \lambda\\
&\geq&0,
\end{eqnarray*}
where the first inequality follows by the fact that the derivative with respect to $z$ is strictly positive for $z\in [\lambda,\mu)$, and the second inequality follows by the second assumption concerning $\kappa$, $\lambda$ and $\mu$. Finally, for $ z\in [0, \lambda)$ we have
\begin{eqnarray*}
\int_z^1 (1-y)g(y)\ud y-(\kappa -1)\mu+\kappa z & =&\int_z^\lambda \kappa \ud y+\int_\lambda^\mu \frac{(\kappa-1)(1-\mu)}{(1-y)^2}\ud y-(\kappa -1)\mu+\kappa z\\
&=&(\kappa-1)(1-\mu)\ln{\frac{1-\lambda}{1-\mu}}-(\kappa -1)\mu+\kappa \lambda\\
&\geq&0,
\end{eqnarray*}
where the inequality follows by the second assumption concerning $\kappa$, $\lambda$ and $\mu$. The proof of the lemma is complete.
\end{proof}

We are now ready to complete the proof of Theorems \ref{thm:cbpoa_main} and \ref{thm.bcce-improved}. The two conditions of Lemma \ref{lem:function} are satisfied for $\kappa=1.7507$, $\lambda=0.225$, and $\mu=0.7966$. By combining Lemmas \ref{lem:ineq-implies-poa-bound}, \ref{lem:f-poa}, and \ref{lem:function}, we conclude that the price of (total) anarchy of GSP auction games over Bayes-Nash equilibria is at most $\frac{\kappa}{(\kappa-1)\mu}< 2.9276$.

\section{Improved Bounds for Pure Nash Equilibria}\label{appendix:fullinfo}
In this section we present our results for pure Nash equilibria in the full information setting (Theorems \ref{thm:full_info_bound} and \ref{few_slots_thm}). For simplicity of exposition, we consider all quality factors to be equal to $1$; so, $\gamma$ does not appear in notation. Our proofs can be adapted to different quality factors in a straightforward way. We consider GSP
auction games with $n$ advertisers with valuations $v_1 \geq \ldots \geq v_n\geq 0$
and $n$ slots with click-through-rates $\alpha_1 \geq \ldots \geq \alpha_n \geq 0$.
We assume that neither all slots have the same click-through-rate nor all
advertisers have the same valuation (in both cases, the price of anarchy is
$1$). 

We use the term {\em inefficiency} of allocation $\pi$ to refer to the ratio
$OPT(\vals)/SW(\pi, \vals)$. In Lemma \ref{theorem_weak_formulation} we showed
that every pure Nash equilibrium corresponds to a weakly feasible allocation.
Hence, the price of anarchy of a GSP auction game over pure Nash equilibria is
upper-bounded by the worst-case inefficiency among weakly feasible allocations.

\begin{defn} An allocation $\pi$ is called {\em proper} if for any two slots $i<j$ with equal click-through-rates, it holds $\pi(i)<\pi(j)$.
\end{defn}
Clearly, for any non-proper weakly feasible allocation, we can construct a proper weakly feasible one with equal social welfare. Hence, in order to prove our upper bounds, we essentially upper-bound the worst-case inefficiency over proper weakly feasible allocations.

Given an allocation $\pi$, consider the directed graph $G(\pi)$ that has one node for each slot, and a directed edge for each advertiser $i$ that connects the node corresponding to slot $i$ to the node corresponding to slot $\pi^{-1}(i)$. In general, $G(\pi)$ consists of a set of disjoint cycles and may contain self-loops.
\begin{defn}
An allocation $\pi$ is called {\em reducible} if its directed graph $G(\pi)$ has more than one cycles. Otherwise, it is called {\em irreducible}.
\end{defn}
Given a reducible allocation $\pi$ such that $G(\pi)$ has $c\geq 2$ cycles, we can construct $c$ GSP auction subgames by considering the slots and the advertisers that correspond to the nodes and edges of each cycle. Similarly, for $\ell=1,\ldots, c$, the restriction $\pi^{\ell}$ of $\pi$ to the slots and advertisers of the $\ell$-th subgame is an allocation for this game. The next fact essentially states that we can focus on irreducible allocations.

\begin{fact}\label{fact:reducible} If allocation $\pi$ is weakly feasible for the original GSP auction game, then $\pi^\ell$ is weakly feasible for the $\ell$-th subgame as well, for $\ell=1,\ldots, c$. Then, the inefficiency of $\pi$ is at most the maximum inefficiency among the allocations $\pi^{\ell}$ for $\ell=1,\ldots, c$.
\end{fact}

When considering irreducible weakly feasible allocations, we further assume that the index of the slot advertiser $1$ occupies is smaller than the index of the advertiser that is assigned to slot $1$. This is without loss of generality due to the following argument. Consider an irreducible weakly feasible allocation $\pi$ for a GSP auction game with $n$ advertisers such that $\pi^{-1}(1)>\pi(1)$. We construct a new game with click-through-rate $a'_i=v_i$ for slot $i$ and valuation $v'_i=\alpha_i$ for advertiser $i$, for $i=1,\ldots, n$, and the allocation $\pi_*=\pi^{-1}$. Observe that $\pi_*^{-1}(1)=\pi(1)<\pi^{-1}(1)=\pi_*(1)$. Clearly, the optimal social welfare is the same in both games while the social welfare of $\pi_*$ for the new game is $SW(\pi_*,\vals') = \sum_i{a'_iv'_{\pi_*(i)}} = \sum_i{v_i \alpha_{\pi_*(i)}}= \sum_i{\alpha_{\pi^{-1}(i)} v_i} = SW(\pi,\vals)$. We can also prove the weak feasibility conditions for $\pi_*$ in the new game for each $i,j$. In order to do so, consider the weak feasibility condition for $\pi$ in the original game for advertisers $\pi(j),\pi(i)$. It is $\alpha_{j} v_{\pi(j)} \geq \alpha_{i}(v_{\pi(j)}-v_{\pi(i)})$ and, equivalently, $v_{\pi(i)} \alpha_{i}  \geq v_{\pi(j)}(\alpha_{i}-\alpha_{j})$. By the definition of the click-through-rates and the valuations in the new game and the definition of $\pi_*$, we obtain that $a'_{\pi_*^{-1}(i)} v'_i \geq a'_{\pi_*^{-1}(j)} (v'_i-v'_j)$ as desired.

We furthermore note that when $v_n=0$, any proper weakly feasible allocation is reducible. This is obviously the case if all advertisers with zero valuation use the last slots. Otherwise, consider an advertiser $i$ with non-zero valuation that is assigned a slot $\pi^{-1}(i)>\pi^{-1}(j)$ where $j$ is an advertiser with zero valuation. Since the allocation is proper, it holds that $\alpha_{\pi^{-1}(i)}<\alpha_{\pi^{-1}(j)}$. Then, we obtain a contradiction by the weak feasibility condition $\alpha_{\pi^{-1}(i)} v_i \geq \alpha_{\pi^{-1}(j)} (v_i-v_j)$ for advertisers $i,j$.

\subsection{GSP auction games with two and three advertisers}\label{sec:3bidders-ub}

We now complete the proof of Theorem \ref{few_slots_thm}.

We begin by presenting the matching upper bound on the price of anarchy for two advertisers and two slots. The upper bound follows by bounding the inefficiency of weakly feasible allocations. Consider a GSP auction game with two slots with click-through-rates $\alpha_1\geq \alpha_2=\beta \alpha_1$, for $\beta\in[0,1]$ and two advertisers with valuations $v_1\geq v_2=\lambda v_1$, for $\lambda\in [0,1]$. The only non-optimal weakly feasible allocation $\pi$ assigns advertiser $1$ to slot $2$ and advertiser $2$ to slot $1$. Its social welfare is $SW(\pi, \vals) = \alpha_1v_2+\alpha_2v_1 = \alpha_1v_1(\beta+\lambda)$, while the optimal social welfare is $OPT(\vals)=\alpha_1v_1+\alpha_2v_2 = \alpha_1v_1(1+\beta\lambda)$. Furthermore, the weak feasibility condition for advertiser $1$ implies that $\alpha_2v_1\geq \alpha_1(v_1-v_2)$, i.e., $\beta\geq 1-\lambda$.
We have that
$$
\frac{OPT(\vals)}{SW(\pi, \vals)} = \frac{1+\beta\lambda}{\beta+\lambda} \leq
\frac{1+(\beta+\lambda)^2/4}{\beta+\lambda} \leq 5/4
$$
where the first inequality holds since the product $\beta\lambda$ is maximized when $\beta=\lambda=(\beta+\lambda)/2$ and the second inequality holds since $\beta+\lambda\in [1,2]$ and the function $\frac{1+x^2/4}{x}$ is non-increasing in $x\in [1,2]$.

For the case of three advertisers, we again present a tight bound on the price of anarchy. We first present the upper bound. Consider a GSP auction game with three slots with click-through-rates $\alpha_1\geq \alpha_2 \geq \alpha_3\geq  0$ and three advertisers with valuations $v_1\geq v_2\geq v_3\geq 0$ and a proper weakly feasible allocation $\pi$ of slots to advertisers. We will prove the theorem by upper-bounding the inefficiency of $\pi$ by $1.259134$. If $\pi$ is reducible, then the inefficiency is bounded by the inefficiency of games with two advertisers (see Fact \ref{fact:reducible}) and the theorem follows by the upper bound of $5/4$ proved for this case. So, in the following, we assume that $\pi$ is irreducible; by the observation above, this implies that $v_3>0$. There are only two such allocations which are in fact symmetric: in the first, slots $1$, $2$, $3$ are allocated to advertisers $3$, $1$, $2$, respectively, and in the second, slots $1$, $2$, $3$ are allocated to advertisers $2$, $3$, $1$, respectively. Without loss of generality (see the discussion above), we assume that $\pi$ is the former allocation.

Let $\beta$, $\delta$, $\lambda$, and $\mu$ be such that $\alpha_2=\beta \alpha_1$, $\alpha_3=\delta \alpha_1$, $v_2=\lambda v_1$, and $v_3=\mu v_1$. Clearly, it holds that $1\geq \beta \geq \delta\geq 0$ and $1\geq \lambda \geq \mu > 0$. The social welfare of allocation $\pi$ is $SW(\pi,\vals)=\alpha_1 v_1 (\mu+\beta+\delta \lambda)$ whereas the optimal social welfare is $OPT(\vals)=\alpha_1 v_1 (1+\beta \lambda+\delta\mu)$. Furthermore, since $\pi$ is weakly feasible, the weak feasibility conditions for advertisers $1$ and $3$ and advertisers $2$ and $3$ are $\alpha_2v_1\geq \alpha_1(v_1-v_3)$ and $\alpha_3v_2\geq \alpha_1(v_2-v_3)$, respectively, i.e., $\beta\geq 1-\mu$ and $\delta\geq 1-\frac{\mu}{\lambda}$.
We are now ready to bound the inefficiency of $\pi$. Let $\epsilon,\theta \geq 0$ be such that $\beta= 1-\mu+\epsilon$ and $\delta=1-\frac{\mu}{\lambda}+\theta$. We have
\begin{eqnarray*}
\frac{OPT(\vals)}{SW(\pi, \vals)} &=&\frac{1+\beta
\lambda+\delta\mu}{\mu+\beta+\delta\lambda} =
\frac{1+\lambda-\mu\lambda+\mu-\frac{\mu^2}{\lambda}+\epsilon\lambda+\theta\mu
}{1+\lambda-\mu+\epsilon+\theta\lambda}\\
&\leq & \frac{1+\lambda-\mu\lambda+\mu-\frac{\mu^2}{\lambda}}{1+\lambda-\mu}.
\end{eqnarray*}
The inequality follows since $1\geq \lambda\geq \mu>0$ implies that $1+\lambda-\mu\lambda+\mu-\frac{\mu^2}{\lambda}=1+\lambda-\mu+\mu(1-\lambda)+\mu(1-\mu/\lambda)\geq 1+\lambda-\mu\geq 1$ and $\epsilon +\theta\lambda \geq \epsilon\lambda+\theta\mu\geq 0$.

For $\mu\in[0,1]$, this last expression is maximized for the value of $\mu$ that makes its derivative with respect to $\mu$ equal to zero, i.e., $\mu=-\sqrt{\lambda^3+1}+\lambda+1$. By substituting $\mu$, we obtain that
$$\frac{OPT(\vals)}{SW(\pi, \vals)} \leq
\frac{\lambda^2+\lambda+2-2\sqrt{\lambda^3+1}}{\lambda}
\leq  1+2\zeta = 1.259134
$$
where $\zeta=0.129567$ and the second inequality follows by the following lemma (Lemma \ref{lem:technical-3}).
\begin{lemma}\label{lem:technical-3}
Let $\zeta=0.129567$. For any $\lambda\in [0,1]$, it holds that $\sqrt{\lambda^3+1}\geq 1-\zeta\lambda+\frac{\lambda^2}{2}$.
\end{lemma}
\begin{proof}
Since both parts of the inequality are non-negative for $\lambda\in[0,1]$, it suffices to show that the function $f(\lambda) = (\lambda^3+1)-\left(1-\zeta\lambda+\frac{\lambda^2}{2}\right)^2$ is non-negative for $\lambda\in [0,1]$. Let $g(\lambda)=-\frac{\lambda^3}{4}+(1+\zeta)\lambda^2-(1+\zeta^2)\lambda+2\zeta$ and observe that $f(\lambda)=\lambda \cdot g(\lambda)$. The proof will follow by proving that $g(\lambda)\geq 0$ when $\lambda\in [0,1]$. Observe that the derivative of $g$ is strictly negative for $\lambda=0$ and strictly positive for $\lambda=1$. Hence, the minimum of $g$
in $[0,1]$ is achieved at the point $\lambda^*=\frac{4+4\zeta-2\sqrt{\zeta^2+8\zeta+1}}{3}$ where the derivative of $g$ becomes zero.
Straightforward calculations yield that $g(\lambda^*)>0$ and the lemma follows.
\end{proof}

In the following we prove that the above analysis is tight. Consider a GSP
auction game with three advertisers with valuations $v_1=1$, $v_2=0.5296$, and
$v_3=0.14583$, respectively, and three slots with click-through-rates
$\alpha_1=1$, $\alpha_2=0.55071$, and $\alpha_3=0.4704$, respectively. Let
$\textbf{b} =(b_1, b_2, b_3)$ be a bid vector with $b_1=0$, $b_2=v_2=0.5296$,
and $b_3=v_3=0.14583$, respectively. So, advertiser $2$ is allocated slot $1$,
advertiser $3$ is allocated slot $2$, and advertiser $1$ is allocated slot $3$.
We refer to this allocation as $\pi$. It is not hard to verify that $\textbf{b}$
is a pure Nash equilibrium, 
\comment{, i.e., when $i>j$ and $v_{\pi(i)}\geq
b_{\pi(j)}$, then $\alpha_i(v_{\pi(i)}-b_{\pi(i+1)})\geq
\alpha_j(v_{\pi(i)}-b_{\pi(j)})$, and when $i\leq j$, then
$\alpha_i(v_{\pi(i)}-b_{\pi(i+1)})\geq \alpha_j(v_{\pi(i)}-b_{\pi(j+1)})$,
assuming that $b_{\pi(4)}=0$. Indeed, it holds that
\begin{eqnarray*}
\alpha_1(v_2-b_3) = 0.38377 &\geq& 0.29166 \approx \alpha_2(v_2-b_1),\\
\alpha_1(v_2-b_3) = 0.38377 &\geq& 0.249124 \approx \alpha_3v_2,\\
\alpha_2(v_3-b_1) \approx 0.08031 &\geq& 0.0686 \approx \alpha_3v_3,\\
\alpha_3v_1 = 0.4704 &\geq& 0.4704 = \alpha_1(v_1-b_2),\\
\alpha_3v_1 = 0.4704 &\geq& 0.4704 \approx \alpha_2(v_1-b_3).
\end{eqnarray*}
By simple calculations, we obtain that the price of anarchy for this allocation
is}
and that the price of anarchy is given by:
\begin{eqnarray*}
\frac{OPT(\vals)}{SW(\pi, \vals)} = \frac{\alpha_1v_1+\alpha_2v_2+\alpha_3v_3}{\alpha_1v_2+\alpha_2v_3+\alpha_3v_1} \geq 1.259133.
\end{eqnarray*}
The proof of Theorem \ref{few_slots_thm} is complete.

\subsection{GSP auction games with many advertisers}\label{sec:many-bidders}
We now prove Theorem \ref{thm:full_info_bound}.
In order to do so, we will actually prove the stronger claim that the worst-case inefficiency among weakly feasible allocations of any GSP auction game is at most $r=\frac{61+7\sqrt{217}}{128}\approx 1.28216$. We use induction. As the base of our induction, we use the fact that GSP auction games with one, two, or three advertisers have worst-case inefficiency among weakly feasible allocations at most $1.28216$. For a single advertiser, the claim is trivial. For two or three advertisers, it follows by the proof of Theorem \ref{few_slots_thm}. Let $n\geq 4$ be an integer. Using the inductive hypothesis that the worst-case inefficiency among weakly feasible allocations of any GSP auction game with at most $n-1$ advertisers is at most $r$, we will show that this is also the case for any GSP auction game with $n$ advertisers.

Consider a GSP auction game with $n$ advertisers with valuations $v_1\geq v_2 \geq \ldots \geq v_n \geq 0$ and $n$ slots with click-through-rates $\alpha_1 \geq \alpha_2 \geq \ldots \geq \alpha_n \geq 0$ and let $\pi$ be a proper weakly feasible allocation. If $\pi$ is reducible, the claim follows by Fact \ref{fact:reducible} and the inductive hypothesis. So, in the following, we assume that $\pi$ is irreducible; this implies that $v_n>0$. Let $j$ be the advertiser that is assigned slot $1$ and $i_1$ be the slot assigned to advertiser $1$. Without loss of generality, we assume that $i_1<j$ since the other case is symmetric; see the discussion at the beginning of Section \ref{appendix:fullinfo}. Also, let $i_2$ be the slot assigned to advertiser $i_1$. By our assumptions, the integers $j$, $1$, $i_1$, and $i_2$ are different.

We will show that
\begin{eqnarray}\label{eq:cost-final}
SW(\pi, \vals)&\geq& \alpha_1v_j+\alpha_{i_1}(v_1-\frac{v_{i_1}}{r})+\alpha_{i_2}(v_{i_1}-v_j)-\frac{\alpha_1v_1}{r}+\frac{OPT(\vals)}{r}.
\end{eqnarray}
Once we have proved inequality (\ref{eq:cost-final}), we can obtain the desired relation between $SW(\pi, \vals)$ and $OPT(\vals)$ using the following technical lemma.
\begin{lemma}\label{lem:technical}
Let $r = \frac{61+7\sqrt{217}}{128} \approx 1.28216$ and $f(\beta, \delta , \lambda, \mu) = \mu+\beta\left(1-\frac{\lambda}{r}\right)+\delta(\lambda-\mu)-\frac{1}{r}$.
Then, the objective value of the mathematical program
\begin{eqnarray*}
\mbox{minimize} & &f(\beta,\delta,\lambda,\mu)\\
\mbox{subject to} & & \beta \geq 1-\mu\\
& & \delta \geq 1-\mu/\lambda\\
& & 1\geq \lambda\geq \mu > 0\\
& & 1\geq \beta, \delta \geq 0
\end{eqnarray*}
is non-negative.
\end{lemma}

\begin{proof}
Since $\mu\leq \lambda\leq 1$, we have that $f(\beta, \delta, \lambda, \mu)$ is non-decreasing in $\beta$ and $\delta$. Using the first two constraints, we have that the objective value of the mathematical program is at least
\begin{eqnarray*}
f\left(1-\mu, 1-\frac{\mu}{\lambda}, \lambda, \mu\right) &=& 1-\frac{1}{r}+\lambda-\frac{\lambda}{r}-\mu\left(2-\frac{\lambda}{r}\right)+\frac{\mu^2}{\lambda},
\end{eqnarray*}
which is minimized for $\mu=\lambda-\frac{\lambda^2}{2r}$ to
\begin{eqnarray*}
f\left(1-\lambda+\frac{\lambda^2}{2r}, \frac{\lambda}{2r}, \lambda, \lambda-\frac{\lambda^2}{2r}\right) &=& 1-\frac{1}{r}-\frac{\lambda}{r}+\frac{\lambda^2}{r}-\frac{\lambda^3}{4r^2}.
\end{eqnarray*}

In order to complete the proof it suffices to show that the function
$g(\lambda) = 1-\frac{1}{r}-\frac{\lambda}{r}+\frac{\lambda^2}{r}-\frac{\lambda^3}{4r^2}$
is non-negative for $\lambda\in [0,1]$.
Observe that $g(\lambda)$ is a polynomial of degree $3$ and, hence, it has at most one local
minimum. Also observe that the derivative of $g(\lambda)$ is
$-\frac{1}{r}+\frac{2\lambda}{r}-\frac{3\lambda^2}{4r^2}$
which is strictly negative for $\lambda=0$ and strictly positive for $\lambda=1$. Hence, its minimum
in $[0,1]$ is achieved at the point $\lambda^*=\frac{4r-2\sqrt{4r^2-3r}}{3}$ where the derivative becomes zero.
Straightforward calculations yield that $g(\lambda^*)=0$ and the lemma follows.
\end{proof}

So, assuming that (\ref{eq:cost-final}) holds, we can apply Lemma \ref{lem:technical} with $\beta=\alpha_{i_1}/\alpha_1$, $\delta=\alpha_{i_2}/\alpha_1$, $\lambda=v_{i_1}/v_1$, and $\mu=v_j/v_1$. Clearly, the last two constraints of the mathematical program in Lemma \ref{lem:technical} are satisfied. Also, observe that the weak feasibility conditions for advertisers $1$ and $j$ and advertisers $i_1$ and $j$ in allocation $\pi$ are $\alpha_{i_1}v_{1}\geq \alpha_1(v_{1}-v_{j})$ and $\alpha_{i_2}v_{i_1}\geq \alpha_{1}(v_{i_1}-v_{j})$, respectively, i.e., $\beta\geq 1-\mu$ and $\delta\geq 1-\mu/\lambda$ and the first two constraints of the mathematical program in Lemma \ref{lem:technical} are satisfied as well. Now, using inequality (\ref{eq:cost-final}) and Lemma \ref{lem:technical}, we have that
\begin{eqnarray*}
SW(\pi, \vals)&\geq & f\left(\frac{\alpha_{i_1}}{\alpha_1}, \frac{\alpha_{i_2}}{\alpha_1}, \frac{v_{i_1}}{v_1}, \frac{v_j}{v_1}\right) \cdot \alpha_1v_1 +\frac{OPT(\vals)}{r} \geq  \frac{OPT(\vals)}{r}
\end{eqnarray*}
and the proof follows.

It remains to prove inequality (\ref{eq:cost-final}). We distinguish between three cases depending on the relative order of $j, i_1$, and $i_2$; in each of these cases, we further distinguish between two subcases. In each case, we exploit the structure of allocation $\pi$ to reason as follows. We consider a restriction of the original game (i.e., a different ``restricted'' game) by removing some advertisers from the original game and the slots they occupy in $\pi$. The particular advertisers to be removed are different in each case. We denote by $\pi'$ the restriction of allocation $\pi$ to the advertisers and slots of the restricted game. We also use $\vals'$ to denote the valuation profile in the restricted game; so, $SW(\pi',\vals')$ denotes the social welfare of $\pi'$ in the restricted game. An important observation is that $\pi'$ is a weakly feasible allocation in the restricted game since the weak feasibility conditions for $\pi'$ are just a subset of the corresponding conditions for $\pi$ (for the original game). Furthermore, the restricted game has at most $n-1$ advertisers and, by the inductive hypothesis, we know that the inefficiency of $\pi'$ is at most $r$. Then, inequality (\ref{eq:cost-final}) follows using this fact and by carefully expressing the optimal social welfare in the new game.

\paragraph{Case I.1: $1<i_1<j<i_2$ and $\alpha_j\leq \alpha_{i_2}r$.}
Consider the restriction of the original game that consists of the advertisers different than $j$, $1$, and $i_1$ and the slots different than $1$, $i_1$, and $i_2$. Let $\pi'$ be the restriction of $\pi$ to the advertisers and slots of the new game and let $\vals'$ be the restriction of $\vals$ to all advertisers besides $j$, $1$ and $i_1$. Clearly, $\pi'$ is weakly feasible for the new game since the weak feasibility conditions for $\pi'$ are just a subset of the corresponding conditions for $\pi$ (for the original game). Also, note that the efficient allocation for the restricted game assigns advertiser $k$ to slot $k$ for $k=2, \ldots, i_1-1, i_1+1, \ldots, j-1, i_2+1,\ldots, n$ and advertiser $k+1$ to slot $k$ for $k=j, \ldots, i_2-1$. By the inductive hypothesis, we know that the inefficiency of $\pi'$ is at most $r$. Hence, we can bound the social welfare of $\pi$ as
\begin{eqnarray*}
SW(\pi, \vals) &= & \alpha_1v_j+\alpha_{i_1}v_1+\alpha_{i_2}v_{i_1}+ \sum_{k\not\in\{1,i_1,i_2\}}{\alpha_k  v_{\pi(k)}}\\
&= & \alpha_1v_j+\alpha_{i_1}v_1+\alpha_{i_2}v_{i_1}+ SW(\pi',\vals')\\
&\geq& \alpha_1v_j+\alpha_{i_1}v_1+\alpha_{i_2}v_{i_1}+ \frac{1}{r}\left(\sum_{k=2}^{i_1-1}{\alpha_kv_k}+\sum_{k=i_1+1}^{j-1}{\alpha_kv_k}+\sum_{k=j}^{i_2-1}{\alpha_kv_{k+1}}+\sum_{k=i_2+1}^{n}{\alpha_kv_k}\right)\\
&\geq & \alpha_1v_j+\alpha_{i_1}v_1+\alpha_{i_2}v_{i_1}+ \frac{1}{r}\left(\sum_{k=2}^{i_1-1}{\alpha_kv_k}+\sum_{k=i_1+1}^{j-1}{\alpha_kv_k}+\sum_{k=j+1}^{i_2}{\alpha_kv_{k}}+\sum_{k=i_2+1}^{n}{\alpha_kv_k}\right)\\
&=& \alpha_1v_j+\alpha_{i_1}v_1+\alpha_{i_2}v_{i_1}+ \frac{1}{r}\left(\sum_{k=1}^{n}{\alpha_kv_k}-\alpha_1v_1-\alpha_{i_1}v_{i_1}-\alpha_{j}v_{j}\right)\\
&\geq & \alpha_1v_j+\alpha_{i_1}(v_1-\frac{v_{i_1}}{r})+\alpha_{i_2}(v_{i_1}-v_j)-\frac{\alpha_1v_1}{r} +\frac{OPT(\vals)}{r}
\end{eqnarray*}
and inequality (\ref{eq:cost-final}) follows. The first inequality follows by the inductive hypothesis and the definition of the efficient allocation for the restricted game. The second inequality follows since $\alpha_k \geq \alpha_{k+1}$ for $k=j, \ldots, i_2-1$. The last inequality follows since $\alpha_j\leq \alpha_{i_2} r$.

\paragraph{Case I.2: $1<i_1<j<i_2$ and $\alpha_j>\alpha_{i_2}r$.}
We use the restriction of the original game that consists of the advertisers different than $j$ and $1$ and the slots different than $1$ and $i_1$.
Now, the efficient allocation for the restricted game assigns advertiser $k$ to slot $k$ for $k=2, \ldots, i_1-1, j+1, \ldots, n$ and advertiser $k-1$ to slot $k$ for $k=i_1+1, \ldots, j$. Using the inductive hypothesis for the restriction $\pi'$ of $\pi$ to the restricted game, we can bound the social welfare of $\pi$ as
\begin{eqnarray*}
SW(\pi, \vals)&= & \alpha_1v_j+\alpha_{i_1}v_1+ \sum_{k\not\in \{1,i_1\}}{\alpha_kv_{\pi(k)}}\\
&= & \alpha_1v_j+\alpha_{i_1}v_1+ SW(\pi',\vals')\\
&\geq & \alpha_1v_j+\alpha_{i_1}v_1+\frac{1}{r}\left(\sum_{k=2}^{i_1-1}{\alpha_kv_k}+\sum_{k=i_1+1}^{j}{\alpha_kv_{k-1}}+\sum_{k=j+1}^{n}{\alpha_kv_{k}}\right)\\
&=& \alpha_1v_j+\alpha_{i_1}v_1+\frac{1}{r}\left(\sum_{k=1}^{n}{\alpha_kv_k}-\alpha_1v_1-\alpha_{i_1}v_{i_1}+\sum_{k=i_1+1}^{j}{\alpha_k(v_{k-1}-v_k)}\right)\\
&\geq & \alpha_1v_j+\alpha_{i_1}v_1+\frac{1}{r}\left(\sum_{k=1}^{n}{\alpha_kv_k}-\alpha_1v_1-\alpha_{i_1}v_{i_1}+\alpha_{j}\sum_{k=i_1+1}^{j}{(v_{k-1}-v_k)}\right)\\
&=& \alpha_1v_j+\alpha_{i_1}v_1-\frac{1}{r} \left(\alpha_1v_1+\alpha_{i_1}v_{i_1}+\alpha_{j}v_{j}-\alpha_{j}v_{i_1}\right)+\frac{OPT(\vals)}{r}\\
&> & \alpha_1v_j+\alpha_{i_1}(v_1-\frac{v_{i_1}}{r})+\alpha_{i_2}(v_{i_1}-v_j)-\frac{\alpha_1v_1}{r} +\frac{OPT(\vals)}{r}
\end{eqnarray*}
and inequality (\ref{eq:cost-final}) follows. The first inequality follows by the inductive hypothesis and the definition of the efficient allocation for the restricted game. The second inequality follows since $\alpha_k \geq \alpha_{j}$ and $v_{k-1}-v_k \geq 0$ for $k=i_1+1, \ldots, j$. The last inequality follows since $\alpha_j> \alpha_{i_2} r$.

\paragraph{Case II.1: $1<i_1<i_2<j$ and $v_{i_2}\leq v_j r$.}
We use the restriction of the original game that consists of the advertisers different than $j$, $1$, and $i_1$ and the slots different than $1$, $i_1$, and $i_2$.
Now, the efficient allocation for the restricted game assigns advertiser $k$ to slot $k$ for $k=2, \ldots, i_1-1, i_1+1, \ldots, i_2-1, j+1, \ldots, n$ and advertiser $k-1$ to slot $k$ for $k=i_2+1, \ldots, j$. Using the inductive hypothesis for the restriction $\pi'$ of $\pi$ to the restricted game, we can bound the social welfare of $\pi$ as
\begin{eqnarray*}
SW(\pi, \vals)&= & \alpha_1v_j+\alpha_{i_1}v_1+\alpha_{i_2}v_{i_1}+ \sum_{k\not\in \{1,i_1,i_2\}}{\alpha_kv_{\pi(k)}}\\
&= & \alpha_1v_j+\alpha_{i_1}v_1+\alpha_{i_2}v_{i_1}+ SW(\pi',\vals')\\
&\geq& \alpha_1v_j+\alpha_{i_1}v_1+\alpha_{i_2}v_{i_1}+ \frac{1}{r}\left(\sum_{k=2}^{i_1-1}{\alpha_kv_k}+\sum_{k=i_1+1}^{i_2-1}{\alpha_kv_k}+\sum_{k=i_2+1}^{j}{\alpha_kv_{k-1}}+\sum_{k=j+1}^{n}{\alpha_kv_k}\right)\\
&\geq& \alpha_1v_j+\alpha_{i_1}v_1+\alpha_{i_2}v_{i_1}+ \frac{1}{r}\left(\sum_{k=2}^{i_1-1}{\alpha_kv_k}+\sum_{k=i_1+1}^{i_2-1}{\alpha_kv_k}+\sum_{k=i_2+1}^{j}{\alpha_kv_{k}}+\sum_{k=j+1}^{n}{\alpha_kv_k}\right)\\
&=& \alpha_1v_j+\alpha_{i_1}v_1+\alpha_{i_2}v_{i_1}+\frac{1}{r} \left(\sum_{k=1}^n{\alpha_kv_k}-\alpha_1v_1-\alpha_{i_1}v_{i_1}-\alpha_{i_2}v_{i_2}\right)\\
&\geq & \alpha_1v_j+\alpha_{i_1}(v_1-\frac{v_{i_1}}{r})+\alpha_{i_2}(v_{i_1}-v_j)-\frac{\alpha_1v_1}{r} +\frac{OPT(\vals)}{r}
\end{eqnarray*}
and inequality (\ref{eq:cost-final}) follows. The first inequality follows by the inductive hypothesis and the definition of the efficient allocation for the restricted game. The second inequality follows since $v_{k-1} \geq v_{k}$ for $k=i_2+1, \ldots, j$. The last inequality follows since $v_{i_2} \leq v_j r$.

\paragraph{Case II.2: $1<i_1<i_2<j$ and $v_{i_2} > v_j r$.}
We use the restriction of the original game that consists of the advertisers different than $1$ and $i_1$ and the slots different than $i_1$ and $i_2$.
Now, the efficient allocation for the restricted game assigns advertiser $k$ to slot $k$ for $k=i_2+1, \ldots, n$, advertiser $i_1+1$ to slot $i_1-1$, and advertiser $k+1$ to slot $k$ for $k=1, \ldots, i_1-2, i_1+1, \ldots, i_2-1$. Using the inductive hypothesis for the restriction $\pi'$ of $\pi$ to the advertisers and slots of the restricted game, we can bound the social welfare of $\pi$ as
\begin{eqnarray*}
SW(\pi, \vals) &= &\alpha_{i_1}v_1+ \alpha_{i_2}v_{i_1}+\sum_{k\not\in \{i_1,i_2\}}{\alpha_kv_{\pi(k)}}\\
&=& \alpha_{i_1}v_1+ \alpha_{i_2}v_{i_1}+SW(\pi',\vals')\\
&\geq & \alpha_{i_1}v_1+\alpha_{i_2}v_{i_1}+\frac{1}{r}\left(\sum_{k=1}^{i_1-2}{\alpha_kv_{k+1}}+\alpha_{i_1-1}v_{i_1+1}+\sum_{k=i_1+1}^{i_2-1}{\alpha_kv_{k+1}}+\sum_{k=i_2+1}^{n}{\alpha_kv_{k}}\right)\\
&=& \alpha_{i_1}v_1+\alpha_{i_2}v_{i_1}+ \frac{1}{r}\left(\sum_{k=1}^{n}{\alpha_kv_k}+\sum_{k=1}^{i_1-2}{(\alpha_k-\alpha_{k+1})v_{k+1}}+(\alpha_{i_1-1}-\alpha_{i_1+1})v_{i_1+1}\right.\\
&&\left.{}+{}\sum_{k=i_1+1}^{i_2-1}{(\alpha_k-\alpha_{k+1})v_{k+1}}-\alpha_1v_1-\alpha_{i_1}v_{i_1}\right)\\
&\geq & \alpha_{i_1}v_1+\alpha_{i_2}v_{i_1}+\frac{1}{r}\left(\sum_{k=1}^{n}{\alpha_kv_k}+\sum_{k=1}^{i_1-2}{(\alpha_k-\alpha_{k+1})v_{i_2}}+(\alpha_{i_1-1}-\alpha_{i_1+1})v_{i_2}\right.\\
&&\left.+\sum_{k=i_1+1}^{i_2-1}{(\alpha_k-\alpha_{k+1})v_{i_2}}-\alpha_1v_1-\alpha_{i_1}v_{i_1}\right)\\
&=& \alpha_{i_1}v_1+\alpha_{i_2}v_{i_1}-\frac{1}{r} \left(\alpha_1v_1+\alpha_{i_1}v_{i_1}+\alpha_{i_2}v_{i_2}-\alpha_{1}v_{i_2}\right)+\frac{OPT(\vals)}{r}\\
&>& \alpha_1v_j+\alpha_{i_1}(v_1-\frac{v_{i_1}}{r})+\alpha_{i_2}(v_{i_1}-v_j)-\frac{\alpha_1v_1}{r} +\frac{OPT(\vals)}{r}
\end{eqnarray*}
and inequality (\ref{eq:cost-final}) follows. The first inequality follows by the inductive hypothesis and the definition of the efficient allocation for the restricted game. The second inequality follows since $\alpha_k-\alpha_{k+1}\geq 0$ and $v_{k+1} \geq v_{i_2}$ for $k=1, \ldots, i_1-2, i_1+1, \ldots, i_2-1$ and $\alpha_{i_1-1}-\alpha_{i_1+1}\geq 0$ and $v_{i_1+1}\geq v_{i_2}$. The last inequality follows since $v_{i_2} > v_j r$ and $\alpha_1>\alpha_{i_2}$.

\paragraph{Case III.1: $1<i_2<i_1<j$ and $v_{i_2}\leq v_j r$.}
We use the restriction of the original game that consists of the advertisers different than $j$, $i_1$, and $1$ and the slots different than $1$, $i_2$, and $i_1$.
Now, the efficient allocation for the restricted game assigns advertiser $k$ to slot $k$ for $k=2, \ldots, i_2-1, j+1, \ldots, n$ advertiser $i_1-1$ to slot $i_1+1$, and advertiser $k-1$ to slot $k$ for $k=i_2+1, \ldots, i_1-1, i_1+2, \ldots, j$. Using the inductive hypothesis for the restriction $\pi'$ of $\pi$ to the restricted game, we can bound the social welfare of $\pi$ as
\begin{eqnarray*}
SW(\pi, \vals) &= & \alpha_1v_j+\alpha_{i_2}v_{i_1}+\alpha_{i_1}v_1+ \sum_{k\not\in \{1,i_2,i_1\}}{\alpha_kv_{\pi(k)}}\\
&= & \alpha_1v_j+\alpha_{i_2}v_{i_1}+\alpha_{i_1}v_1+ SW(\pi',\vals')\\
&\geq& \alpha_1v_j+\alpha_{i_2}v_{i_1}+\alpha_{i_1}v_1+ \frac{1}{r}\left(\sum_{k=2}^{i_2-1}{\alpha_kv_k}+\sum_{k=i_2+1}^{i_1-1}{\alpha_kv_{k-1}}+\alpha_{i_1+1}v_{i_1-1}\right.\\
&&\left.+\sum_{k=i_1+2}^{j}{\alpha_kv_{k-1}}+\sum_{k=j+1}^{n}{\alpha_kv_k}\right)\\
&\geq& \alpha_1v_j+\alpha_{i_2}v_{i_1}+\alpha_{i_1}v_1+ \frac{1}{r}\left(\sum_{k=2}^{i_2-1}{\alpha_kv_k}+\sum_{k=i_2+1}^{i_1-1}{\alpha_kv_{k}}+\sum_{k=i_1+1}^{n}{\alpha_kv_{k}}\right)\\
&=& \alpha_1v_j+\alpha_{i_2}v_{i_1}+\alpha_{i_1}v_1+ \frac{1}{r}\left(\sum_{k=1}^{n}{\alpha_kv_k}-\alpha_1v_1-\alpha_{i_2}v_{i_2}-\alpha_{i_1}v_{i_1}\right)\\
&\geq & \alpha_1v_j+\alpha_{i_1}(v_1-\frac{v_{i_1}}{r})+\alpha_{i_2}(v_{i_1}-v_j)-\frac{\alpha_1v_1}{r} +\frac{OPT(\vals)}{r}
\end{eqnarray*}
and inequality (\ref{eq:cost-final}) follows. The first inequality follows by the inductive hypothesis and the definition of the efficient allocation for the restricted game. The second inequality follows since $v_{k-1}\geq v_k$ for $k=i_2+1, \ldots, i_1-1, i_1+2, \ldots, j$ and $v_{i_1-1}\geq v_{i_1+1}$. The last inequality follows since $v_{i_2} \leq v_j r$.

\paragraph{Case III.2: $1<i_2<i_1<j$ and $v_{i_2} > v_j r$.}
We use the restriction of the original game that consists of the advertisers different than $i_1$ and $1$ and the slots different than $i_2$ and $i_1$.
Now, the efficient allocation for the restricted game assigns advertiser $k$ to slot $k$ for $k=i_2+1, \ldots, i_1-1, i_1+1, \ldots, n$ and advertiser $k+1$ to slot $k$ for $k=1, \ldots, i_2-1$. Using the inductive hypothesis for the restriction $\pi'$ of $\pi$ to the restricted game, we can bound the social welfare of $\pi$ as
\begin{eqnarray*}
SW(\pi, \vals) &= & \alpha_{i_2}v_{i_1}+\alpha_{i_1}v_1+ \sum_{k\not\in \{i_2,i_1\}}{\alpha_kv_{\pi(k)}}\\
&= & \alpha_{i_2}v_{i_1}+\alpha_{i_1}v_1+ SW(\pi',\vals')\\
&\geq & \alpha_{i_2}v_{i_1}+\alpha_{i_1}v_1+\frac{1}{r}\left(\sum_{k=1}^{i_2-1}{\alpha_kv_{k+1}}+\sum_{k=i_2+1}^{i_1-1}{\alpha_kv_{k}}+\sum_{k=i_1+1}^{n}{\alpha_kv_{k}}\right)\\
&=& \alpha_{i_2}v_{i_1}+\alpha_{i_1}v_1+\frac{1}{r}\left(\sum_{k=1}^{n}{\alpha_kv_k}-\alpha_1v_1-\alpha_{i_1}v_{i_1}+\sum_{k=1}^{i_2-1}{(\alpha_k-\alpha_{k+1})v_{k+1}}\right)\\
&\geq & \alpha_{i_2}v_{i_1}+\alpha_{i_1}v_1+\frac{1}{r}\left(\sum_{k=1}^{n}{\alpha_kv_k}-\alpha_1v_1-\alpha_{i_1}v_{i_1}+\sum_{k=1}^{i_2-1}{(\alpha_{k}-\alpha_{k+1})v_{i_2}}\right)\\
&=& \alpha_{i_2}v_{i_1}+\alpha_{i_1}v_1-\frac{1}{r} \left(\alpha_1v_1+\alpha_{i_1}v_{i_1}+\alpha_{i_2}v_{i_2}-\alpha_{1}v_{i_2}\right)+\frac{OPT(\vals)}{r}\\
&> & \alpha_1v_j+\alpha_{i_1}(v_1-\frac{v_{i_1}}{r})+\alpha_{i_2}(v_{i_1}-v_j)-\frac{\alpha_1v_1}{r} +\frac{OPT(\vals)}{r}
\end{eqnarray*}
and inequality (\ref{eq:cost-final}) follows. The first inequality follows by the inductive hypothesis and the definition of the efficient allocation for the restricted game. The second inequality follows since $\alpha_k-\alpha_{k+1}\geq 0$ and $v_{k+1}\geq v_{i_2}$ for $k=1, \ldots, i_2-1$. The last inequality follows since $v_{i_2} > v_j r$ and $\alpha_1>\alpha_{i_2}$.

The proof of Theorem \ref{thm:full_info_bound} is complete.

\section{Improved Bounds for Learning Outcomes}\label{appendix:learning}
In this section, we focus on the full information game. For simplicity we assume that all quality factors $\gamma_i=1$, and assume that players are sorted so that $v_1\ge v_2\ge \dots \ge v_n$ (all proofs extend to the case with general quality factors by considering effective values $\gamma_iv_i$ in place of valuations everywhere).

The main goal of this Appendix is to prove Theorem \ref{thm.cce}. Similarly to the proof of Theorem \ref{thm:cbpoa_main} in Appendix \ref{appendix:bayes} for Bayes-Nash equilibria, the proof considers a player $i$ with valuation $v_i$, possible bids of the form $yv_i$, and uses the fact that the player has no-regret about such alternative bids. In the full information case, we can handle the player with top valuation separately, and will only use that this player 1 has no regret about bidding her actual valuation $v_1$. For any other player $i$, the proof is analogous to the proof of Theorem \ref{thm:cbpoa_main} in Appendix \ref{appendix:bayes}. However, we no longer have to consider separately the case when the player's optimal slot is 1. This allows us to drop one requirement for the function $g$ in the definition \ref{def:bounded}. We further simplify that definition by setting $\delta=\beta$ (we have verified that different values for $\delta$ do not yield any improvement). More formally, we will need the following definition.

\begin{defn}\label{def:cor-bounded}
Let $\beta \in (0,1]$. A function $g:[0,1]\rightarrow \mathbb{R}_+$ is called $\beta$-bounded if the following two properties hold:
\begin{eqnarray*}
i)&& \int_0^1 g(y)\ud y \leq 1,\\
ii)&&  \int_z^1 (1-y)g(y)\ud y \geq \beta -(1+\beta) z, \quad \forall z\in[0,1].
\end{eqnarray*}
\end{defn}

The following lemma states the connection of the price of anarchy to the existence of $\beta$-bounded functions.

\begin{lemma}\label{lem:cor-f-poa}
Let $\beta\in (0,1]$ be such that a $\beta$-bounded function exists. Then, the price of total anarchy of the Generalized Second Price auction in the full information setting is at most $1+1/\beta$.
\end{lemma}

\begin{proof}
In the proof, we consider a GSP auction game with $n$ slots with click-through-rates $\alpha_1\geq \alpha_2 \geq \ldots \geq \alpha_n\geq 0$ and $n$ conservative players with valuations $v_1, v_2, \ldots, v_n\geq 0$. Let $\mathbf{b}$ denote the bids of the players at a coarse correlated equilibrium.

We begin by lower-bounding the expected utility of each player at a coarse correlated equilibrium. We first consider player $1$ and her deviation to the bid $v_1$. Then, player $1$ would always be allocated slot $1$ and would pay the highest bid among the remaining players (which is at most $b_{\pi(1)}$) per click. By the definition of the coarse correlated equilibrium such a deviation does not increase her expected utility (as the player has no regret), i.e.,
\begin{eqnarray}\label{eq:bidder-1}
\mathbb{E}[u_1(\bids)] \geq \mathbb{E}[u_1(v_1, \bids_{-1})]\geq
\E[\alpha_1(v_1-b_{\pi(1)})] \geq \beta \alpha_1v_1 - \beta
\mathbb{E}[\alpha_1b_{\pi(1)}],
\end{eqnarray}
where the last inequality follows since $v_1\geq b_{\pi(1)}$ and since $\beta \in (0,1]$.
Now, consider the deviation of player $i$ to the deterministic bid $b'_i\leq v_i$. Then, she would be assigned to slot $i$ or higher and would get utility at least $\alpha_i(v_i-b'_i)$ when the $i$-th highest bid is smaller than $b'_i$. Again, by the definition of the coarse correlated equilibrium such a deviation does not increase her expected utility, i.e.,
\begin{eqnarray*}
\mathbb{E}[u_i(\bids)]\geq \mathbb{E}[u_i(b_i',\bids_{-i})] \geq \E[\alpha_i (v_i-b_i')\one{b_{\pi(i)}<b_i'}].
\end{eqnarray*}
Using the first property in Definition \ref{def:cor-bounded} for $g$ as well as the last inequality (with $b'_i=yv_i$) , we have
\begin{eqnarray*}
\mathbb{E}[u_i(\bids)] &\geq & \int_0^1{g(y)\cdot \mathbb{E}[u_i(\bids)] \ud y}\\
&\geq & \int_0^1{g(y)\cdot \E[\alpha_i (v_i-yv_i)\one{b_{\pi(i)}<yv_i}]\ud y}\\
&=& \E[\alpha_iv_i\int_{0}^1{(1-y)g(y)\one{b_{\pi(i)}<yv_i}\ud y}]\\
&=& \E[\alpha_iv_i \int_{b_{\pi(i)}/v_i}^1{(1-y)g(y)\ud y}].
\end{eqnarray*}
We now apply the second property of Definition \ref{def:cor-bounded} for function $g$ to obtain
\begin{eqnarray}\label{ineq:corr_utility}
\mathbb{E}[u_i(\bids)] \geq \E[\beta \alpha_iv_i - (1+\beta)\alpha_i b_{\pi(i)}]= \beta \alpha_iv_i - (1+\beta)\E[\alpha_ib_{\pi(i)}].
\end{eqnarray}

By summing over all players and using inequalities (\ref{eq:bidder-1}) and (\ref{ineq:corr_utility}), we have
\begin{eqnarray*}
\sum_i\mathbb{E}[u_i(\bids)]&=&\E[u_1(\bids)]+\sum_{i\geq 2}\E[u_i(\bids)]\\
&\geq& \beta \sum_i \alpha_i v_i - (1+\beta) \sum_i \E[\alpha_i b_{\pi(i)}]+ \E[\alpha_1 b_{\pi(1)}]\\
&=&\beta OPT(\vals)- (1+\beta) \sum_i \E[\alpha_i b_{\pi(i)}]+ \E[\alpha_1 b_{\pi(1)}].
\end{eqnarray*}

Now, we use this last inequality in the same way we used inequality (\ref{semi-smooth-weak}) in the proof of Lemma \ref{lem:ineq-implies-poa-bound}. By the
fact that the social welfare is the sum of the expected utilities of the players plus the total payments, we obtain
\begin{eqnarray*}
\E[SW(\pi(\bids),\vals)]&=& \mathbb{E}[\sum_i{u_i(\bids)}] +\mathbb{E}[\sum_i{\alpha_ib_{\pi(i+1)}}]\\
&\geq& \beta OPT(\vals)- (1+\beta) \sum_i \E[\alpha_i b_{\pi(i)}]+ \E[\alpha_1 b_{\pi(1)}]+\sum_{i\geq 2}\E[{\alpha_ib_{\pi(i)}}]\\
&=& \beta OPT(\vals)- \beta \sum_i \E[\alpha_i b_{\pi(i)}]\\
&\geq&\beta OPT(\vals)- \beta \E[SW(\pi(\bids),\vals)],
\end{eqnarray*}
which implies that the price of total anarchy $OPT(\vals)/\E[SW(\pi(\bids),\vals)]$ is at most $1+1/\beta$, as desired.
\end{proof}

We are ready to complete the proof of Theorem \ref{thm.cce}. By Lemma \ref{lem:cor-f-poa}, it suffices to find a $\beta$-bounded function with $\beta$ as high as possible. Let $\lambda\approx 0.4328$ be the solution of the equation $1-\lambda+\ln{(1-\lambda)}=0$ and $g:[0,1]\rightarrow \R_+$ be the function defined as follows:
\begin{eqnarray*}
g(y) = \left\{
\begin{array}{ll}
\frac{1}{(1-\lambda)(1-y)},& y \in [0,\lambda]\\
0,& y\in(\lambda, 1]
\end{array} \right.
\end{eqnarray*}
We will show that $g$ is $\beta$-bounded for $\beta=\frac{\lambda}{1-\lambda}$; the upper bound of $1/\lambda\approx 2.3102$ stated in Theorem \ref{thm.cce} will then follow.

Indeed, by the definition of $\lambda$, we have
\begin{eqnarray*}
\int_0^1 g(y)\ud y = \int_0^\lambda{\frac{\ud y}{(1-\lambda)(1-y)}}= -\frac{\ln(1-\lambda)}{1-\lambda} = 1.
\end{eqnarray*}
Hence, $g$ satisfies the first property of Definition \ref{def:cor-bounded}. We also observe that
\begin{eqnarray*}
\int_z^1{(1-y)g(y)\ud y} = \int_{\min\{z,\lambda\}}^\lambda{\frac{\ud y}{1-\lambda}}\geq \frac{\lambda-z}{1-\lambda} = \beta-(1+\beta)z,
\end{eqnarray*}
i.e., $g$ satisfies the second property of Definition \ref{def:cor-bounded} as well.

\section{Irrational and Partially Rational Players}\label{appendix:irrational}

In this section we consider the effect of partial rationality on the welfare generated
by the GSP auction.  We first consider a setting in which the players are not necessarily
perfect utility optimizers, but rather can only be assumed to apply strategies that
form an approximate equilibrium.  We then study a setting in which some fraction of
the players bid arbitrarily, without any rationality assumptions beyond avoiding the
dominated strategy of overbidding (see Section \ref{subsec:nonoverbidding}).
In both cases, we find that the social welfare guarantees of the GSP auction
degrade continuously with the degree of irrationality present in the players.

\subsection{Approximate equilibria}

We will consider the social welfare generated by the GSP auction with uncertainty when players play only approximately utility-maximizing strategies.  In Section \ref{sec:uncertainty}, we assumed that rational players apply strategies at equilibrium.  However, due to limits on rationality or indifference between small differences in utility, it may be the case that players converge only to an approximate equilibrium.  We begin by defining this notion formally.
Given a joint distribution $(\dists,\textbf{G})$ over types and quality factors, we say that strategy profile $\bids$ is an $\epsilon$-Bayes-Nash equilibrium for distributions
$\dists,\textbf{G}$ if, for all players $i$, all types $\vali$, and all alternative strategies $\bidi'$,
\[
\E_{\valsmi,\gamma,\bids}[u_i(\bidi(\vali),\bidsmi(\valsmi), \gamma) \vert v_i ]
\geq
(1-\epsilon)\E_{\valsmi,\gamma,\bids }[\utili(\bid'_i(\vali), \bidsmi(\valsmi), \gamma) \vert v_i].
\]
Notice our choice of the multiplicative definition of approximate equilibria, justified by the fact that we have chosen not to scale values to lie in $[0,1]$.

We define the \emph{$\epsilon$-Bayes-Nash Price of Anarchy} to be
$$\sup_{\dists, \mathbf{G}, \bids(\cdot) \in \epsilon\text{-}BNE} \frac{\E_{\vals, \gamma}
[OPT(\vals,\gamma)]}{\E_{ \vals,\gamma,
\bids(\vals)}[SW(\pi(\bids(\vals),\gamma),\vals,\gamma)]}$$
where $\epsilon$-BNE is the set of all $\epsilon$-Bayes-Nash equilibria.

We now claim that our bound for social welfare at (non-approximate) equilibrium
\comment{Theorem \ref{thm:cbpoa_main}} degrades
continuously as we relax the degree to which a bidding strategy only
approximates an equilibrium.

\begin{theorem}
\label{thm:approx-equil}
The $\epsilon$-Bayes-Nash price of anarchy of the Generalized Second Price auction is at most $1.2553 + (1-\epsilon)^{-1}\cdot 1.6722$.
\end{theorem}

The intuition behind Theorem \ref{thm:approx-equil} is that the bound for
exact equilibria obtained in Theorem \ref{thm:cbpoa_main} depends on the
Bayes-Nash equilibrium condition in a continuous way.
This continuity is captured by the semi-smoothness
of the GSP auction (as well as by inequality \eqref{semi-smooth-weak},
used to prove Theorem \ref{thm:cbpoa_main} in Appendix \ref{appendix:bayes}).
Indeed, the following Lemma follows by a trivial modification to the proof of
Lemma \ref{lem:ineq-implies-poa-bound}.

\begin{lemma}
Assume that for every GSP auction game there is a bidding profile $\bids'$ and parameters $\beta, \delta>0$ such that inequality (\ref{semi-smooth-weak}) holds for any strategy profile $\bids$. Then the $\epsilon$-Bayes-Nash price of anarchy of the Generalized Second Price auction is at most $\frac{(1-\epsilon)^{-1}+\delta}{\beta}$.
\end{lemma}

It then follows immediately from
Lemmas \ref{lem:f-poa} and \ref{lem:function} (see Appendix \ref{appendix:bayes})
that the $\epsilon$-Bayes-Nash
price of anarchy of the GSP auction is at most
$\frac{(1-\epsilon)^{-1}+0.7507}{0.7507 \cdot 0.7966} \approx 1.2553 + (1-\epsilon)^{-1}\cdot 1.6722$.

\subsection{Irrational players}

We now consider a setting in which, of the $n$ advertisers who bid in the GSP
auction, some subset of them are ``irrational'' and cannot be assumed to
apply strategies at equilibrium.  We still think of the irrational
advertisers as being true players in the GSP auction, with valuations and quality scores.
The irrational advertisers simply may not apply rational bidding strategies; for
example, they may not have experience with the GSP auction, or not know about
historical bidding patterns.

Our setting will be an extension of the GSP auction with uncertainty.  We will
first provide some definitions.  Given valuations $\vals$, quality scores $\gamma$,
an outcome $\pi$, and a set $S$ of players, the social welfare restricted to set $S$
is $SW_S(\pi,\vals,\gamma) = \sum_{i \in S} \alpha_{\pi^{-1}(i)} \gamma_i \val_{i}$,
the total value of the outcome $\pi$ for the advertisers in $S$.  The optimal social
welfare restricted to $S$ is $OPT_S(\vals, \gamma) = \max_{\pi} SW_S(\pi,\vals,\gamma)$.

Given a joint distribution $(\dists,\textbf{G})$ over types and quality factors, and a set $S$ of players,
we say that strategy profile $\bids$ is an $S$-Bayes-Nash equilibrium for distributions
$\dists,\textbf{G}$ if, for all players $i \in S$, all types $\vali$, and all alternative strategies $\bidi'$,
\[
\E_{\valsmi,\gamma,\bids}[u_i(\bidi(\vali),\bidsmi(\valsmi), \gamma) \vert v_i ]
\geq
\E_{\valsmi,\gamma,\bids }[\utili(\bid'_i(\vali), \bidsmi(\valsmi), \gamma) \vert v_i].
\]
That is, no player in $S$ can improve her utility by modifying her bid, but no such restriction is
imposed upon the players outside $S$.

\comment{
Given an outcome $\pi$ (which is an assignment of these $n+m$ advertisers to $n+m$
slots), the definition of social welfare is unchanged: it is
$SW(\pi,\vals,\gamma) = \sum_{i \in N \cup M} \alpha_{\sigma(i)} \gamma_i
v_i$.  We define the social welfare of advertisers in $N$ to be precisely
that: $SW_N(\pi,\vals,\gamma) = \sum_{i \in N} \alpha_{\sigma(i)} \gamma_i
v_i$.  The optimal
social welfare for advertisers in $N$ is $OPT_N(\vals, \gamma) = \max_{\pi}
SW_N(\pi,\vals,\gamma)$.
}

We will show that, for each set $S$ of players, the total expected social welfare
obtained by GSP at an $S$-Bayes-Nash equilibrium is a good
approximation to $\E[OPT_S(\vals)]$.  We can interpret this result as stating
that the addition of irrational players does not significantly degrade the social welfare
that would have been generated had they not participated.  Note that we cannot hope to
always obtain a good approximation to $\E[OPT(\vals)]$ (the optimal social
welfare of \emph{all} advertisers) at all $S$-Bayes-Nash equilibria; for example, it may
be that the valuations of the players outside $S$ are very large, but they choose
(irrationally) to bid $0$.

We note that our no-overbidding assumption (Section \ref{subsec:nonoverbidding})
continues to apply to all players, not only to the players in $S$.  In other words,
we require that $\bidi(\vali) \leq \vali$ for all $i \not\in S$ and all $\vali$.  We
feel this is a natural restriction to impose even on ``irrational'' advertisers,
as overbidding is an easily-avoided dominated strategy.  Moreover, it is arguable
that inexperienced advertisers would bid conservatively, and not risk a large
payment with no gain.\footnote{This relies on the simplifying assumption that
all advertisers have knowledge of their own private valuations.  Admittedly, this
requires a certain level of sophistication and may be difficult to attain in
practice.  Our argument is thus limited to imperfect strategy choice given
perfect knowledge of types.  It remains open to extend this analysis to players
who may misunderstand their own valuations.}

Formally, given a non-empty subset $S$ of advertisers, we define the \emph{$S$-Bayes-Nash Price of Anarchy} to be
$$\sup_{\dists, \mathbf{G}, \bids(\cdot) \in S\text{-}BNE} \frac{\E_{ \vals, \gamma}
[OPT_S(\vals,\gamma)]}{\E_{ \vals,\gamma,
\bids(\vals)}[SW(\pi(\bids(\vals),\gamma),\vals,\gamma)]}$$
where $S$-BNE is the set of all $S$-Bayes-Nash equilibria.

Our main result is the following extension of Theorem \ref{thm:cbpoa_main}.
\begin{theorem}\label{thm:irrational}
For any non-empty subset $S$ of rational advertisers, the $S$-Bayes-Nash price of anarchy
of the Generalized Second Price auction is at most $2.927$.
\end{theorem}
In order to prove this theorem, we need an inequality similar to inequality (\ref{semi-smooth-weak}) in Section \ref{appendix:bayes}. In particular, for every bid profile $\bids$, there exists a bid profile $\bids'$ defined on the rational advertisers such that

\begin{equation}\label{semi-smooth-weak-irrational}
\sum_{i\in S}{\E[u_i(b'_i(v_i),\bidsmi,\gamma)]} \geq \beta
\E[OPT_S(\vals,\gamma)] - (1+\delta) \sum_{i\in
S}{\E[\alpha_{\nu(i)}\gamma_{\pi(\nu(i))}b_{\pi(\nu(i))}]}
+ \E[\alpha_1\gamma_{\pi(1)}b_{\pi(1)}].
\end{equation}

We can prove inequality (\ref{semi-smooth-weak-irrational}) by following the same steps as in the proof of Lemma \ref{lem:f-poa} and by considering the utilities of the rational players at their most profitable deviation. Here, $\nu(i)$ should be interpreted as the slot the rational advertiser $i$ occupies in the efficient allocation restricted to $S$ and $\pi(j)$ is the advertiser that occupies the $j$-th slot in allocation $\pi$ (this advertiser can be rational or irrational). Similarly, $\pi^i(j)$ is the player with the $j$-th highest effective bid among all advertisers besides the rational advertiser $i$.

All the arguments hold in this case as well. However, there is a minor point that should be justified. Observe that in order to obtain inequalities (\ref{eq:bn-improved-slot1}) and (\ref{eq:bn-improved-slotj}), we used the fact that the $j$-th highest effective bid (excluding advertiser $i$) is not larger that the effective value of advertiser $i$ when $\nu(i)=j$, i.e., when slot $j$ is allocated to advertiser $i$ in the efficient allocation restricted to $S$. When adapting the proof to the case of rational and irrational players, it may be the case that $\nu(i)=j$ when the rational advertiser $i$ has valuation $v_i=x$ but $\gamma_{\pi^i(j)}b_{\pi^i(j)}>\gamma_i x$. This may be due to the fact that player $\pi^i(j)$ is one of the irrational players. Fortunately, both inequalities (\ref{eq:bn-improved-slot1}) and (\ref{eq:bn-improved-slotj}) are obviously true in this case as well. Observe that $\beta\leq \delta$ (otherwise, the second property of Definition \ref{def:bounded} would not hold for $z=1$) and the right-hand side of both inequalities is non-positive. The changes in the rest of the proof of Lemma \ref{lem:f-poa} are minor.

Then, Theorem \ref{thm:irrational} follows by the next lemma that exploits inequality (\ref{semi-smooth-weak-irrational}) and using the same values for $\beta$ and $\delta$ that we used in Section \ref{appendix:bayes}.
\begin{lemma}\label{lem:ineq-implies-poa-bound-irrational}
Assume that for every GSP auction game with a non-empty set $S$ of rational players there is a bidding profile $\bids'$ for the players in $S$ and parameters $\beta, \delta>0$ such that inequality (\ref{semi-smooth-weak-irrational}) holds for any strategy profile $\bids$. Then, the $S$-Bayes-Nash price of anarchy of the Generalized Second Price auction is at most $\frac{1+\delta}{\beta}$.
\end{lemma}

\begin{proof}
Consider an $S$-Bayes-Nash equilibrium $\bids$. Define $\bids'$ as in inequality (\ref{semi-smooth-weak-irrational}) and observe that $\E[u_i(\bids,\gamma)] \geq \E[u_i(b'_i(v_i),\bidsmi,\gamma)]$ for every player $i\in S$. We use this inequality and the fact that the social welfare is at least the sum of the expected utilities of the rational advertisers plus the total payments to get
$$\begin{aligned}
& \E[SW(\pi(\mathbf{b(v)},\gamma),\mathbf{v},\gamma)] \geq  \sum_{i\in
S}{\mathbb{E}[u_i(\mathbf{b},\gamma)]}
+\sum_i{\mathbb{E}[\alpha_i\gamma_{\pi(i+1)}b_{\pi(i+1)}]} \\
&\qquad \geq \sum_{i\in S}{\E[u_i(b'_i(v_i), \bidsmi,\gamma)]} + \sum_{i\geq
2}{\E[\alpha_i\gamma_{\pi(i)}b_{\pi(i)}]}\\
&\qquad \geq \beta\E[OPT_S(\mathbf{v},\gamma)]-(1+\delta)\sum_{i\in S}
\E[\alpha_{\nu(i)}\gamma_{\pi(\nu(i))}b_{\pi(\nu(i))}]+\E[\alpha_1\gamma_{\pi(1)
}b_{\pi(1)}]+\sum_{i\geq 2}{\E[\alpha_i\gamma_{\pi(i)}b_{\pi(i)}]}\\
&\qquad \geq \beta\E[OPT_S(\mathbf{v},\gamma)]-(1+\delta)\sum_{i}
\E[\alpha_{i}\gamma_{\pi(i)}b_{\pi(i)}]+\sum_{i}{\E[\alpha_i\gamma_{\pi(i)}b_{
\pi(i)}]}\\
& \qquad =
\beta\E[OPT_S(\mathbf{v},\gamma)]-\delta\sum_i{\E[\alpha_i\gamma_{\pi(i)}b_{
\pi(i)}]}\\
&\qquad \geq
\beta\E[OPT_S(\mathbf{v},\gamma)]-\delta\E[SW(\pi(\mathbf{b(v)},\gamma),\mathbf{
v},\gamma)],
\end{aligned}$$
which implies that the $S$-Bayes-Nash price of anarchy is at most $\frac{1+\delta}{\beta}$, as desired.
\end{proof}

\end{document}